\pdfoutput=1
\documentclass[12pt]{article}
\usepackage[left=1.0in,right=1.0in,top=1.0in,bottom=1.0in]{geometry}
\usepackage{setspace}
\usepackage{fancyhdr}
\usepackage{lmodern}
\usepackage{microtype}
\usepackage[normalem]{ulem}
\usepackage[T1]{fontenc}

\usepackage{titlesec}
\usepackage[title]{appendix}
\usepackage{titling}
\pretitle{\begin{center}\large\bfseries}
\posttitle{\end{center}}
\preauthor{\begin{center}\normalsize}
\postauthor{\end{center}}
\predate{\begin{center}\normalsize}
\postdate{\end{center}}

\usepackage{amssymb, amsmath, amsfonts, mathtools}
\usepackage{dsfont}
\usepackage{amsthm}

\newtheorem{proposition}{Proposition}

\newtheorem{assumption}{Assumption}

\theoremstyle{remark}
\newtheorem{remark}{Remark}
\theoremstyle{plain}

\newcommand{\E}{\mathbb{E}}

\newcommand{\N}{\mathcal{N}}

\newcommand{\1}[1]{\mathds{1}\left[#1\right]}

\usepackage{array, booktabs, makecell, cellspace}
\usepackage{siunitx}
\usepackage[flushleft]{threeparttable}
\usepackage{rotating, tabularx}

\usepackage{graphicx, subcaption}
\graphicspath{{figures/}}
\usepackage{pdflscape, tikz}
\usetikzlibrary{decorations.pathreplacing, arrows.meta, positioning}
\usepackage{caption}
\usepackage{float}

\usepackage{enumitem}

\usepackage{xurl}
\usepackage{xcolor}
\definecolor{citationcolor}{RGB}{0, 127, 255}

\usepackage[authoryear,round]{natbib}
\setcitestyle{aysep={,}}
\newcolumntype{L}[1]{>{\raggedright\let\newline\\\arraybackslash\hspace{0pt}}m{#1}}
\newcolumntype{C}[1]{>{\centering\let\newline\\\arraybackslash\hspace{0pt}}m{#1}}
\newcolumntype{R}[1]{>{\raggedleft\let\newline\\\arraybackslash\hspace{0pt}}m{#1}}

\usepackage[hidelinks, breaklinks, colorlinks=true,
            linkcolor=citationcolor, citecolor=citationcolor,
            urlcolor=citationcolor]{hyperref}

\title{Agentic Delegation and the Language Frontier of Software Developers:
A Model and Evidence from Claude Code on GitHub\thanks{Preliminary draft, July 2026.
Comments welcome. An earlier version circulated as ``Coding Beyond Your
Training: Claude Code and the Technological Frontier of Software
Developers.'' This research was made possible by Alexander Quispe's time
as a research associate at Microsoft--GitHub during summer 2026.}}

\author{Alexander Quispe\thanks{Department of Computing and Mathematical
Sciences, California Institute of Technology; Microsoft--GitHub.
Email: aquisper@caltech.edu.}
\qquad
Kevin Xu\thanks{GitHub. Email: khxu@github.com.}}

\date{\today}

\begin{document}

\maketitle
\thispagestyle{empty}

\begin{abstract}
\noindent \singlespacing
We develop and test a model of agentic delegation in software production.
Developers face language-specific entry thresholds; conversational AI mainly
augments work in languages they already know, while agentic AI adds delegated
execution under developer specification and verification. The model predicts
an activation band of unfamiliar languages that become feasible only with an
agent, expanding the developer's observed language-production frontier. We
test this prediction in a monthly GitHub panel of 5{,}346 developers, dating
adoption by first Claude Code co-authorship and constructing commit-level
language outcomes from 57 million changed files. Doubly robust
staggered-adoption event studies with not-yet-treated comparisons show sharp
expansion at adoption: active languages rise by 2.5 relative to a 0.9 baseline,
newly used languages by 1.2, entropy by 0.38, and cumulative breadth continues
to grow afterward. The pattern survives removing the treatment-defining
language, excluding all Claude-coauthored commits, conditioning on activity,
and screening users of competing agents. Consistent with the model, first uses
of unfamiliar languages concentrate among narrow pre-adoption specialists at
each activity level. Because adoption is voluntary and may coincide with
project shocks, the estimates are event-time associations rather than
definitive causal effects.
\end{abstract}

\vspace{1em}

\vspace{0.5em}

\newpage
\setcounter{page}{1}

\section{Introduction}
\label{sec:introduction}

A skilled Python developer rarely writes production Rust. A data scientist
trained on R and SQL rarely ships an iOS application in Swift. The boundaries
of an individual's programming-language portfolio are sticky, set by
human-capital investments made during training and early career, and costly to
cross even when the underlying problem-solving ability would transfer. These
boundaries matter beyond the individual: they prevent skilled people from
contributing to projects that need their general abilities but happen to be
coded in unfamiliar languages.

This paper studies whether agentic coding assistants expand the observed
programming-language portfolios of individual developers. Not all AI
assistants work through the same margin. A first generation of conversational
and co-pilot tools supplies suggestions, explanations, snippets, and debugging
advice that the developer must read, adapt, and integrate --- assistance that
is most valuable where she already has a language-specific foothold. Claude
Code represents a second generation: an agentic assistant that can inspect a
repository, edit files, run commands, debug failures, and iterate from a
natural-language specification. In the model we develop, this adds a
\emph{delegation} mode to the developer's production menu: the agent executes
part of the language-specific work while the developer specifies, decomposes,
and verifies. Delegation lowers the entry threshold for some
unfamiliar-language opportunities, creating an \emph{activation band} of
languages that would not have been used under solo or conversational
production but become feasible with an agent. This yields predictions for the
monthly language count, newly-used languages, and cumulative language breadth,
which we test in a monthly GitHub panel around first detectable Claude Code
use.

The frontier we measure deserves emphasis, because it defines the claim. It is
a \emph{production} frontier, not a skill frontier: the set of languages in
which a developer ships working code, under whichever production mode she
selects. We do not claim, and the model does not predict, that a Python
developer learns to write Rust unassisted; we claim that she can now
\emph{deliver} a Rust project by directing and verifying an agent --- and that
this margin, new with the second generation of tools, is economically real.
Delegated output is not autonomous machine output: its value in the model is
zero for a developer who cannot specify and verify, and empirically the
developer selects the project, steers the work, and ships the result.

Measurement matches the question. Because GitHub attaches no language to a
commit, we rebuild each sampled developer's contribution history at the commit
level: we enumerate 3.2 million commits across 133{,}000 developer-repository
pairs --- discovered from the contribution API, the raw Claude-commit stream,
and the GH~Archive event record, which adds 45 percent more repositories than
the contribution API alone --- and extract the 57 million files they changed,
assigning each file a language with GitHub's Linguist rules and keeping only
programming languages. Treatment is dated by the developer's first
Claude-co-authored commit, detected through the machine-readable co-author
trailer, from a universe of 7.8 million such commits. The analysis sample
contains 5{,}346 developers observed over 28 months: early adopters
(April--September 2025) and later adopters (October 2025--January 2026) who
serve as not-yet-treated comparisons before their own adoption, screened for
automation accounts and for detectable use of competing agentic tools before
adoption. We estimate group-time average treatment effects with the doubly
robust \citet{callaway2021difference} estimator, allowing one month of
anticipation because the trailer-based date is a slightly late measure of true
first use.

Around adoption, language portfolios expand sharply. The number of distinct
programming languages a developer works in rises by 2.5 in the adoption month,
against a pre-adoption mean of 0.9, and remains significantly elevated for
months. Newly-used languages --- languages absent from the developer's entire
observed history, the direct counterpart of the activation band --- spike by
1.2 at adoption, roughly four times the baseline flow. Language entropy, which
requires the new languages to carry non-trivial shares of the work, rises by
0.38. The cumulative stock of languages ever used grows by 1.6 at adoption and
keeps growing through the post-adoption window, the stock-flow pattern the
model predicts; we treat this outcome as descriptive, however, because it is
the one outcome with non-trivial pre-period coefficients, indicating that
early adopters were already accumulating languages somewhat faster before
adoption. Repository counts and commit volume also jump at adoption; we report
them as diagnostics of engagement, not as frontier outcomes.

The central measurement concern is mechanical: adoption is dated by a commit,
and commits carry languages, so the adoption-month estimate could in principle
reflect the treatment-defining commit itself, amplified by the surge in
activity that accompanies adoption. Four families of checks address this
directly. Excluding, for every developer and every month, the language her
first Claude commit introduced leaves the effects significant at all horizons.
Excluding \emph{all} Claude-co-authored commits from outcome construction ---
so that outcomes are built only from the developer's unassisted work ---
preserves roughly two-thirds of the language-count effect, which remains
significant for months; the expansion is not an artifact of counting the
agent's output. Conditioning on pre-adoption activity leaves the estimates
unchanged, while per-commit rate outcomes show that the persistent component
of the expansion operates through a larger volume of work spanning more
languages rather than through each commit becoming more polyglot. Stricter
activity filters and alternative timing assumptions do not change the picture.

The model's sharpest cross-sectional prediction also finds support. Expected
expansion into unfamiliar languages is the product of headroom --- the stock
of unfamiliar candidates --- and the ability to specify and verify. Double
sorting developers on pre-adoption activity and pre-adoption portfolio
breadth, first uses of unfamiliar languages concentrate among
\emph{specialists} at every ability level: developers with narrow pre-adoption
portfolios begin two to three times more new languages at adoption than
generalists of the same activity level. A uniform activity shock would not
produce this pattern.

We are explicit about the central identification threat. Claude adoption is
voluntary, and its timing is plausibly correlated with the decision to start a
project in an unfamiliar language: a developer who decides to ship a Rust
application may install Claude \emph{because} the language is unfamiliar,
making her first Claude commit mechanically contemporaneous with her first
Rust commit. The staggered design addresses cohort heterogeneity and
negative-weighting concerns
\citep{goodmanbacon2021difference,dechaisemartin2020two}, and the outcome-side
checks above rule out the mechanical explanations, but no outcome-side check
resolves project-driven selection into adoption timing. We therefore interpret
the estimates as event-time associations around first detectable Claude Code
use --- sharp, robust, and quantitatively consistent with the
agentic-delegation mechanism --- rather than as settled causal effects, and we
outline the designs that could settle them.

The paper contributes to three literatures. To the empirical economics of AI
in knowledge work
\citep{peng2023impact,brynjolfsson2024generative,noy2023experimental,
cui2025effects,dellacqua2023jagged}, which has so far studied first-generation
completion and conversational tools, we bring the second generation:
commit-message co-authorship identifies \emph{agentic} production at scale in
observational data, over horizons long enough to see dynamics, with the
outcome shifted from the speed of a fixed task to the breadth of the
technological portfolio. To the literature on automation, augmentation, and
the task content of work
\citep{acemoglu2019automation,idetalamas2025,agrawal2018prediction}, we add an
individual-level extensive margin: a delegation technology that expands the
set of tasks a given worker can produce in, distinct from displacement and
reinstatement. To the user-innovation tradition
\citep{vonhippel2005democratizing,baldwin2011modeling,boudreau2013innovation},
our developer-level evidence speaks to how agentic tools restructure the
feasible set of open-source contributors.

The remainder of the paper is organized as follows. Section
\ref{sec:literature} positions the paper in the literature. Section
\ref{sec:background} documents the two-generation history of AI coding
assistance with primary sources. Section
\ref{sec:theory} develops the agentic-delegation model and its empirical
predictions. Section \ref{sec:data} describes the data sources, sample
construction, and outcome variables. Section \ref{sec:strategy} sets out the
empirical strategy. Section \ref{sec:results} reports the main results.
Section \ref{sec:robustness} documents robustness to mechanical
exposure, activity volume, and alternative sample restrictions. Section
\ref{sec:mechanism} tests the model's specialist-heterogeneity
prediction. Section \ref{sec:discussion} confronts the identification threats
and outlines a research agenda for cleaner causal tests. Section
\ref{sec:conclusion} concludes.

\section{Related Literature}
\label{sec:literature}

The paper connects three literatures: the theory of automation versus
augmentation, from which we take the two-generation distinction; models
of human capital, learning, and technology adoption, from which we take
the entry-barrier machinery; and the empirical economics of AI-assisted
knowledge work, which has so far studied first-generation tools. A
fourth, on open-source and user innovation, supplies the setting.

\subsection{Automation versus augmentation: two generations of AI}

The closest formal antecedent to our two-generation framework is
\citet{idetalamas2025}, who model AI that converts compute into agents
operating either \emph{autonomously} (co-workers that perform tasks) or
\emph{non-autonomously} (co-pilots that assist a human), inside a
knowledge hierarchy with endogenous occupational choice. Their central
prediction --- autonomous AI primarily benefits the most knowledgeable
workers, while non-autonomous AI benefits the least knowledgeable ---
is the vertical counterpart of the distinction we develop horizontally,
across a developer's language portfolio. More broadly, the task-based
automation framework of \citet{acemoglu2018artificial,
acemoglu2019automation, acemoglu2025simple} separates machines
\emph{taking over} tasks from technologies that create or reinstate
human tasks, and \citet{agrawal2018prediction} locate the
augment-versus-automate threshold in the division between machine
prediction and human judgment. We import that division directly: in our
delegation mode the agent executes while the developer retains
specification and verification.

Two recent contributions microfound the delegation mode itself.
\citet{catalini2026agi} model the AI transition as a race between a
falling cost of automated execution and a bounded human capacity to
\emph{verify} machine output, so that as execution costs collapse, the
binding constraint becomes verification --- precisely the role played by
the verification cost $\kappa(a,s)$ in our model, which we make
decreasing in the developer's own ability. \citet{gans2026oring} embed
automation in an O-ring production structure \citep{kremer1993oring},
rationalizing why a single unfamiliar-language component can cap the
value of an entire project --- and why an agent that raises the floor on
that weak link changes what projects are feasible. At the macro end,
\citet{hampole2025ai} summarize AI's labor-market impact through
occupation-level task exposure, and explicitly flag worker-level skill
acquisition and new tasks as outside their model; that individual-level,
extensive margin --- which languages a given developer can produce in ---
is exactly the object our model and data speak to.

\subsection{Human capital, learning, and technology adoption}

The entry barrier at the center of our model descends from
\citet{jovanovic1996learning}: expertise is human capital accumulated by
Bayesian learning, switching to a new technology destroys part of it,
and the loss grows with the size of the technological leap --- so a
sufficiently expert agent may optimally never switch. In our setting the
agentic assistant bears precisely this barrier on the developer's
behalf. The horizontal skill primitive comes from
\citet{gibbons1999wage} and especially the task-specific human capital
of \citet{gibbons2004task}: skill accumulated on a particular task,
portable across employers but forfeited on switching tasks, is the
natural formalization of \emph{language-specific} capital, and its
additive role in output is the insertion point for both generations of
assistance. The Bayesian machinery follows the learning-about-fixed-
characteristics tradition surveyed by \citet{baleyveldkamp2021}.

On the adoption margin, \citet{ulu2009adoption} show in a dynamic
information-acquisition model that a lower fixed adoption cost expands
the adoption region, raises the probability of eventual adoption, and
shortens time-to-adoption; our activation band is the static analogue,
with delegation playing the role of the cost reduction for
unfamiliar-language opportunities. \citet{pastorino2024careers}
integrates learning, assignment, and human capital in a framework where
optimal assignment solves a bandit with dependent arms --- the natural
engine for the follow-on question of \emph{which} unfamiliar language a
developer tries first, which we leave to future work.

A tension in this literature motivates our heterogeneity analysis: the
first-generation empirical evidence finds the largest gains for the
least skilled (below), while \citet{idetalamas2025} predict autonomous
AI favors the most knowledgeable. Our specialist results (Section
\ref{sec:mechanism}) suggest the resolution is horizontal rather than
vertical: what matters for the extensive language margin is not skill
level but \emph{headroom} --- the stock of unfamiliar candidates ---
interacted with the general ability to specify and verify.

\subsection{Empirical effects of AI on knowledge work and code}

A rapidly growing empirical literature measures the productivity effects
of generative AI. \citet{noy2023experimental} find large time savings in
professional writing tasks, concentrated among the less skilled;
\citet{brynjolfsson2024generative} find a 14--15 percent average
productivity gain for customer-support agents, rising to 34 percent for
novices; \citet{dellacqua2023jagged} document a ``jagged frontier'' of
task-level help and harm in consulting. In software specifically,
\citet{peng2023impact} report 56 percent faster task completion in a
GitHub Copilot RCT, and \citet{cui2025effects} find 26 percent average
output gains in field experiments, larger for short-tenure developers.
\citet{anthropic2025skills} and \citet{metr2025measuring} provide early
evidence specific to Claude-assisted and experienced open-source
workflows.

Every study in this literature examines a \emph{first-generation} tool
--- completion or conversational assistance --- and none distinguishes
agentic from conversational use in observational data. Meanwhile the
capability of autonomous agents is documented on the engineering side:
on SWE-bench \citep{jimenez2024swebench}, agents receive a repository
and an issue and must produce a resolving patch end-to-end, including in
multilingual settings. Our contribution is to bring this second
generation into the economics: the Claude Code co-authorship trailer
identifies agentic production at scale in public commit data, the
28-month panel admits dynamics that short experiments cannot, and the
outcome is the \emph{technological frontier} --- which languages a
developer produces in --- rather than the speed of a fixed task.

\subsection{User innovation and open-source software}

The observation that ordinary users innovate, and that open communities
scale this innovation, is associated with
\citet{vonhippel1988sources, vonhippel2005democratizing};
\citet{baldwin2011modeling} formalize the shift from producer to open
collaborative innovation, and \citet{lakhani2003hackers} and
\citet{boudreau2013innovation} document the motivations and outputs of
open-source developers. Our results bear on this literature's feasible
set: if agentic delegation expands the language portfolios of individual
contributors as we document, the matching of contributors to projects
--- and hence the structure of cumulative open-source innovation ---
changes with it. The companion peer-effects work \citep{conti2024peer}
examines the network dimension explicitly.

\section{Background: Two Generations of AI Coding Assistance}
\label{sec:background}

The two-generation taxonomy that organizes the theory is not a labeling
choice; it is a datable fact about the technology. This section
establishes the timeline with primary sources, because three features of
it do identifying work in the empirical design.\footnote{A documented
history with the full chronology and primary-source citations is
available as a background appendix. Dates cited here were verified
against the original announcements.} Throughout, the two generations are
defined by the \emph{locus of execution}, not by model quality:
Generation-1 tools propose text --- a completion, a snippet, an
explanation --- that the developer must read, adapt, place, and run
herself; Generation-2 tools operate the computer --- they open and edit
files across a repository, execute commands, run tests, read the errors,
and iterate --- while the developer specifies and verifies.

\emph{Generation 1 (2019--2023).} LLM code completion predates ChatGPT
by three years (Deep TabNine, 2019). GitHub Copilot --- ``suggesting
whole lines or entire functions'' for the developer to accept --- was
previewed in June 2021 and generally available by June 2022, powered by
the Codex model whose paper also introduced the era's benchmark,
HumanEval: a self-contained function graded by unit tests, with no
repository context, no command execution, and no iteration
\citep{friedman2021copilot,dohmke2022copilotga,chen2021codex}. ChatGPT
(November 2022) added the conversational form
\citep{openai2022chatgpt}, and by mid-2023 a developer had abundant
Generation-1 assistance --- Copilot, CodeWhisperer, Tabnine, Ghostwriter,
ChatGPT, Claude, Bard --- none of which could open a second file, run a
test, or read an error message unprompted.

\emph{The transition (2023).} Within months of GPT-4, open-source
projects (AutoGPT, GPT-Engineer, Aider) wrapped chat models in loops
with file access --- delegation \emph{mechanics} without delegation
\emph{capability}. SWE-bench quantified the gap: given a real GitHub
issue and its repository, the best model of the Generation-1 era
resolved 1.96 percent of tasks \citep{jimenez2024swebench}. However good
chat models were at suggestion-shaped work, they could not do
repository-level software engineering.

\emph{Generation 2 (2024--2025).} No widely used tool crossed from
suggestion to delegation before March 2024, when Devin
\citep{cognition2024devin} and the open-source SWE-agent --- which framed
language models as ``a new category of end users'' of software,
operating it through purpose-built interfaces \citep{yang2024sweagent}
--- arrived within weeks of each other. Capability then improved by more
than an order of magnitude in twenty months: from 1.96 percent
(Claude~2, late 2023) to 33 percent (GPT-4o, August 2024) to 49 percent
(Claude~3.5, October 2024) to above 70 percent (Claude Opus~4, May 2025)
on human-validated repository tasks
\citep{openai2024swebenchverified,anthropic2024computeruse}. Mass-market
delegation arrived in a narrow window: agentic IDE modes in November
2024 (Cursor, Windsurf), Copilot agent mode in February 2025
\citep{github2025agentmode}, \textbf{Claude Code} --- ``an active
collaborator that can search and read code, edit files, write and run
tests, commit and push'' --- in research preview on February 24, 2025
and general availability on May 22, 2025
\citep{anthropic2025claudecode}, OpenAI's Codex agent in May 2025
\citep{openai2025codex}, and Google's Gemini CLI in June 2025
\citep{google2025geminicli}. Appendix \ref{app:tools} classifies each
tool into the model's production modes.

Three features of this timeline matter for the design. First,
\emph{Generation 1 was universally available throughout the
pre-period}: ChatGPT and Copilot predate our earliest adoption cohort
(April 2025) by more than two years, so the pre-adoption counterfactual
is solo work plus mature conversational assistance --- exactly the menu
$\mathcal{M}_1$ of the model --- and event-time changes around Claude
Code adoption cannot reflect gaining access to conversational AI.
Second, \emph{the generational boundary is datable and narrow}:
mass-market delegation became available between November 2024 and
May 2025, so our adoption cohorts (April--September 2025) sit
immediately after the modality arrived --- the right place to measure an
extensive-margin response. Third, \emph{competing agentic tools bracket
the treatment}: Devin, Cursor's agent, Copilot agent mode, Codex, and
Gemini CLI overlap Claude Code's rollout, which is why Section
\ref{sec:data} screens out developers with detectable competing-agent
use before adoption --- for the remaining sample, first detectable
Claude Code use marks entry into the delegation mode itself.

\section{A Model of Agentic AI and Language-Frontier Expansion}
\label{sec:theory}

This section presents the core of the model; derivations, proofs, and
extensions are collected in Appendix \ref{app:proofs}. The key distinction
is between two generations of AI assistance. A first-generation assistant
augments work in languages the developer already understands: it suggests
code, explains errors, and accelerates implementation, but the developer
must still read and integrate the output. A second-generation assistant
adds delegation: it can inspect files, write code, run commands, and debug
failures from a natural-language specification. The empirically relevant
comparison is therefore not ``no AI'' versus ``AI.'' Conversational
assistants were already available before the Claude Code rollout. The
comparison is a menu expansion from solo-plus-augmentation to
solo-plus-augmentation-plus-delegation.

\subsection{Environment and production modes}

There are developers $i$, programming languages $k\in\{1,\ldots,K\}$, and
months $t$. Developer $i$ has a familiar-language set $\mathcal{K}_i$ and
an unfamiliar set $\mathcal{U}_i$ with $U_i\equiv|\mathcal{U}_i|$; a
specialist is a developer with small $|\mathcal{K}_i|$ and therefore many
unfamiliar candidates. Empirically, ``unfamiliar'' means not observed in
the developer's prior public GitHub contributions; it need not imply
that the developer has no private experience with the language. For each developer--language pair, $s_{ik,t}\in[0,1]$
is language-specific execution skill, and the developer holds Normal
beliefs $\theta_{ik}\sim\N(\mu_{ik,t},1/\pi_{ik,t})$ about her latent
productivity match, with low skill and low precision in unfamiliar
languages. General ability $a_i\ge0$ --- the capacity to specify,
decompose, and verify a task --- is language-invariant. Each language-month
carries an opportunity value $\omega_{ik,t}$ and an activation cost
$b_{ik,t}$. The developer has CARA utility with coefficient $\rho$, so a
payoff with mean $m$ and variance $\sigma^2$ has certainty equivalent
$m-\rho\sigma^2/2$ (Appendix \ref{app:proofs}).

Suppressing $(i,k,t)$ subscripts, the three modes deliver net
certainty-equivalent surpluses
\begin{align}
V^S &= \omega+s\mu-\frac{\rho s^2}{2\pi}-b,
\label{eq:solo-surplus}\\
V^C &= V^S+\gamma s-r_C,
\label{eq:copilot-surplus}\\
V^D &= \omega+(1-\lambda)s\mu+\lambda a z(A)
-\kappa(a,s)-r_D-b
-\frac{\rho}{2}
\left[
\frac{(1-\lambda)^2s^2}{\pi}
+\sigma_D^2(a,s,A)
\right].
\label{eq:delegate-surplus}
\end{align}
Solo production exposes the developer to uncertainty about her own match.
Generation-1 augmentation yields a deterministic gain $\gamma s$
proportional to existing skill, net of interaction cost $r_C$: the
developer benefits most when she can read, adapt, and verify the
suggestion herself. Generation-2 delegation hands a share
$\lambda\in(0,1]$ of execution to an agent with competence $z(A)$,
increasing in capability $A$ and scaled by the developer's general
ability; it carries verification cost $\kappa(a,s)$, compute cost $r_D$,
and residual error variance $\sigma_D^2(a,s,A)$, while shrinking the
variance attached to the human match from $s^2/\pi$ to
$(1-\lambda)^2 s^2/\pi$.

\begin{assumption}[Augmentation requires a foothold]
\label{ass:foothold}
For unfamiliar languages with skill $\underline{s}$,
$\gamma\underline{s}-r_C\le0$. For familiar languages with skill $\bar{s}$,
$\gamma\bar{s}-r_C>0$.
\end{assumption}

\begin{assumption}[Verification technology]
\label{ass:verification}
Verification costs and residual error are weakly lower for stronger
developers, more familiar languages, and more capable agents:
$\kappa_a<0$, $\kappa_s\le0$, and $\partial\sigma_D^2/\partial a\le0$,
$\partial\sigma_D^2/\partial s\le0$, $\partial\sigma_D^2/\partial A\le0$.
\end{assumption}

Assumption \ref{ass:foothold} says Generation 1 is valuable on the
intensive margin of known languages but does not by itself make an
unfamiliar language viable on impact.

\subsection{Menus, thresholds, and the activation band}

The pre-agent menu is $\mathcal{M}_1=\{S,C\}$; agentic adoption expands it
to $\mathcal{M}_2=\{S,C,D\}$. Under generation $g$ the developer takes the
best available mode, $V^g_{ik,t}=\max_{m\in\mathcal{M}_g}V^m_{ik,t}$, and
language $k$ is \emph{active} when $Z^g_{ik,t}=\1{V^g_{ik,t}\ge0}$, so the
monthly language count is $N^g_{it}=\sum_k Z^g_{ik,t}$.

Because each surplus is affine in the opportunity shock $\omega$, each
mode has an activation threshold. Setting $V^S\ge0$ in Equation
\eqref{eq:solo-surplus}, solo production is viable when
\begin{equation}
\omega\ge
T^S\equiv b-s\mu+\frac{\rho s^2}{2\pi},
\label{eq:solo-threshold}
\end{equation}
and since $V^C=V^S+\gamma s-r_C$, augmentation is viable when
$\omega\ge T^C=T^S-(\gamma s-r_C)$. The developer takes whichever
pre-agent mode clears first, so the effective Generation-1 threshold is
\begin{equation}
T^1=\min\{T^S,T^C\}
=T^S-\max\{0,\gamma s-r_C\},
\label{eq:gen1-threshold}
\end{equation}
and for an unfamiliar language Assumption \ref{ass:foothold} implies
$\gamma s-r_C\le0$, hence $T^1=T^S$: augmentation does not move the
entry margin. Setting $V^D\ge0$ in Equation \eqref{eq:delegate-surplus},
delegation is viable when
\begin{align}
\omega\ge T^D
\equiv{}&
b-(1-\lambda)s\mu-\lambda a z(A)
+\kappa(a,s)+r_D
\nonumber\\
&\quad
+\frac{\rho}{2}
\left[
\frac{(1-\lambda)^2s^2}{\pi}
+\sigma_D^2(a,s,A)
\right],
\label{eq:delegate-threshold}
\end{align}
so the post-agent threshold is $T^2=\min\{T^1,T^D\}\le T^1$.

For an unfamiliar language, the economics reduce to the \emph{agentic
threshold reduction} $B\equiv T^1-T^D=T^S-T^D$. Differencing Equations
\eqref{eq:solo-threshold} and \eqref{eq:delegate-threshold} (algebra in
Appendix \ref{app:proofs}) gives
\begin{equation}
B
=
\lambda\bigl[a z(A)-s\mu\bigr]
-\kappa(a,s)-r_D
+\frac{\rho}{2}
\left[
\frac{(2\lambda-\lambda^2)s^2}{\pi}
-\sigma_D^2(a,s,A)
\right].
\label{eq:agentic-advantage}
\end{equation}
The first term is expected execution substitution: agentic output
replaces a share $\lambda$ of the human's uncertain language-specific
output with agent execution $a z(A)$. The second is the verification and
compute cost of delegating. The third is risk substitution: delegation
shrinks the variance attached to the human match but introduces residual
agent-error risk. If $B>0$, delegation lowers the entry threshold for
the unfamiliar language, and by Assumption \ref{ass:verification} the
reduction is larger for higher-ability developers and more capable
agents.

Menu expansion alone delivers a weak frontier result:

\begin{proposition}[Frontier expansion]
\label{prop:frontier}
For every developer, language, date, and opportunity realization,
$Z^2_{ik,t}\ge Z^1_{ik,t}$, hence $N^2_{it}\ge N^1_{it}$ path by path.
\end{proposition}

Proposition \ref{prop:frontier} is nearly mechanical --- the developer can
always ignore the delegation option. The economically substantive
prediction is \emph{where} delegation changes behavior:

\begin{proposition}[Activation band for unfamiliar languages]
\label{prop:band}
Consider an unfamiliar language satisfying Assumption \ref{ass:foothold}.
If $B_{ik,t}>0$, then
\begin{equation}
Z^2_{ik,t}-Z^1_{ik,t}
=
\1{T^D_{ik,t}\le\omega_{ik,t}<T^S_{ik,t}}.
\label{eq:band-indicator}
\end{equation}
If the conditional opportunity CDF $F_{ik,t}$ is continuous, the
probability that delegation activates the language is
$F_{ik,t}(T^S_{ik,t})-F_{ik,t}(T^D_{ik,t})$, and the expected
language-count expansion is
\begin{equation}
\E\!\left[N^2_{it}-N^1_{it}\right]
=
\sum_{k}
\left[
F_{ik,t}\!\left(T^1_{ik,t}\right)
-
F_{ik,t}\!\left(T^2_{ik,t}\right)
\right]
\ge 0.
\label{eq:expected-count}
\end{equation}
\end{proposition}

Delegation does not need to make every opportunity attractive. It
activates the \emph{middle band}: opportunities too weak for solo or
conversational production but strong enough once the agent can execute
part of the task. The band is the model's central object, and the
newly-used-language outcome in the empirical analysis is its direct
counterpart. Figure \ref{fig:band} illustrates, placing the two
historical tool generations of Section \ref{sec:background} on the
opportunity line.

\begin{figure}[H]
\centering
\begin{singlespace}
\begin{tikzpicture}[scale=0.95, every node/.style={font=\small}]
\draw[->, thick] (0,0) -- (13.2,0)
  node[right, font=\small\itshape] {$\omega_{ik,t}$};
\fill[blue!14] (4.6,0) rectangle (8.8,1.35);
\draw[very thick] (4.6,-0.18) -- (4.6,1.55);
\node[above right] at (4.62,1.5) {$T^D_{ik,t}$};
\draw[very thick] (8.8,-0.18) -- (8.8,1.55);
\node[above right] at (8.82,1.5) {$T^1_{ik,t}=T^S_{ik,t}$};
\node[align=center] at (2.3,0.72)
  {inactive under both menus\\[-1pt] $Z^1=Z^2=0$};
\node[align=center] at (6.7,0.72)
  {\textbf{activation band}\\[-1pt] $Z^1=0,\ Z^2=1$};
\node[align=center] at (11.1,0.72)
  {active under both menus\\[-1pt] $Z^1=Z^2=1$};
\node[align=left, anchor=south west, text width=5.5cm, blue!40!black]
  (gentwo) at (0.0,3.05)
  {Generation 2 (Claude Code, Devin, Codex, Gemini CLI, \ldots) creates
   mode $D$: entry now at $T^D < T^1$};
\draw[-{Stealth}, thick, blue!60!black] (3.2,3.0) -- (4.5,1.65);
\node[align=left, anchor=south west, text width=5.9cm, gray!30!black]
  (genone) at (7.0,3.05)
  {Generation 1 (ChatGPT, Copilot, \ldots) does \emph{not} move this
   threshold for unfamiliar $k$: $\gamma s - r_C \le 0$
   (Assumption \ref{ass:foothold}), so $T^1 = T^S$};
\draw[-{Stealth}, thick, gray!70!black] (8.1,3.0) -- (8.72,1.65);
\draw[decorate, decoration={brace, mirror, amplitude=6pt}]
  (4.6,-0.35) -- (8.8,-0.35);
\node[align=center, below] at (6.7,-0.62)
  {languages used \emph{only} because of delegation\\
   $Z^2_{ik,t}-Z^1_{ik,t} = \1{T^D_{ik,t}\le
   \omega_{ik,t} < T^1_{ik,t}}$\\
   $\Rightarrow$ observed as \emph{newly-used languages} at adoption};
\end{tikzpicture}
\end{singlespace}
\caption{The activation band for an unfamiliar language $k$, with the
tool generations of Section \ref{sec:background} placed on it.
Generation-1 tools leave the entry threshold at $T^1=T^S$ because their
value is proportional to skill the developer does not have in $k$;
Generation-2 tools add delegation, whose threshold $T^D$ does not
require language-specific skill. Opportunities falling in $[T^D, T^1)$
are produced only under menu $\mathcal{M}_2$ --- the model's prediction
for the spike in newly-used languages at first detectable Claude Code
use.}
\label{fig:band}
\end{figure}
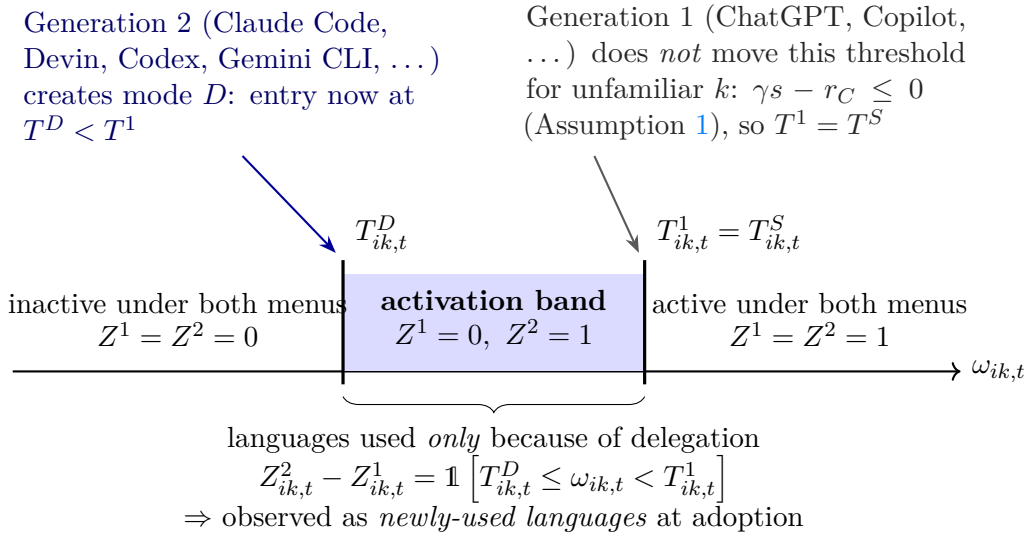

\begin{remark}[What the model does and does not claim]
\label{rem:production-frontier}
The frontier in this model is a \emph{production} frontier, not a
\emph{skill} frontier. A language enters the active set when the
developer can profitably deliver work in it under \emph{some} mode in
her menu --- including delegation, in which the agent executes share
$\lambda$ while the developer specifies, decomposes, and verifies. The
model does not claim that adoption raises language-specific skill
$s_{ik,t}$: a developer who could only write Python solo may now ship a
Rust project, but by directing an agent, not by coding Rust unassisted.
Delegated output is nonetheless not autonomous machine output --- its
value $\lambda a_i z(A)$ is zero for a developer with no general
ability, and it carries the human verification cost $\kappa(a,s)$.
Skill acquisition enters only through the learning extension in
Appendix \ref{app:proofs}, which we treat as a secondary channel.
\end{remark}

\subsection{Flows versus stocks}

The band prediction concerns first uses, which are a \emph{flow}: a given
language can enter the portfolio only once. The corresponding \emph{stock}
is the cumulative count of languages ever used, and the two move
differently after adoption.

\begin{proposition}[Dynamic cumulative-language effect]
\label{prop:dynamic}
For an initially unfamiliar language, let $p^g_{ik}$ be the per-period
first-use hazard under generation $g$. If $p^2_{ik}\ge p^1_{ik}$, the
expected cumulative-language effect at event-time horizon $s$ is
\begin{equation}
\Delta C_i(s)
=
\sum_{k\in\mathcal{U}_i}
\left[
(1-p^1_{ik})^{s+1}-(1-p^2_{ik})^{s+1}
\right]\ge0,
\label{eq:cumulative-gap}
\end{equation}
which in the closed-frontier benchmark $p^1_{ik}=0<p^2_{ik}$ is strictly
increasing and concave over the observed horizon.
\end{proposition}

Newly-used languages may therefore spike at adoption and revert as the
at-risk set is depleted, while the cumulative stock keeps rising --- the
signature pattern the event studies test. Two extensions are developed in
Appendix \ref{app:proofs}: expected expansion is largest for high-ability
specialists, who combine many unfamiliar candidates with low verification
costs (Proposition \ref{prop:specialist}), and Bayesian learning from
agentic interaction can convert initially delegated languages into
persistent portfolio members by raising belief precision.

\subsection{Empirical mapping}

The model maps directly to the outcomes in Section \ref{sec:data}. The
monthly language count and language entropy capture the active-language
set of Equation \eqref{eq:expected-count}. Newly-used languages isolate
the activation band of Proposition \ref{prop:band}. Cumulative languages
correspond to the stock dynamics of Proposition \ref{prop:dynamic}.
Repository and commit counts are activity diagnostics rather than
frontier outcomes: they help establish that language expansion reflects
substantive new work, and they feed the volume-based robustness checks.
The specialist-heterogeneity prediction motivates a test by pre-adoption
portfolio breadth and ability proxies, carried out in Section
\ref{sec:mechanism}.

\section{Data}
\label{sec:data}

The empirical analysis combines two ingredients assembled from public
GitHub data: the universe of Claude-co-authored commits between January
2025 and January 2026, which dates each developer's first detectable use
of Claude Code, and a newly constructed \emph{commit-level} record of
every sampled developer's public contribution history over a 28-month
window (January 2024 through April 2026), in which the programming
language of each commit is measured from the files it changed. We
describe the treatment data, the sample, the commit-level outcome
construction, and the resulting analysis panel in turn.

\subsection{Source: Claude-co-authored commits}

Claude Code, released in beta in February 2025 and at general availability
in May 2025, attaches a co-author trailer of the form
\texttt{Co-Authored-By: Claude} to every commit it produces.\footnote{The
trailer is a Git convention by which any number of additional authors can
be associated with a commit; co-authors appear in the commit metadata and
are surfaced on GitHub. The Claude trailer carries a model identifier,
e.g.\ \texttt{claude-opus-4} or \texttt{claude-sonnet-4}, which we use
descriptively but do not exploit for identification.} Because this trailer
is machine-readable and propagates to GitHub's public event API, every
commit produced with Claude assistance can be detected and aggregated at
scale. We harvest 7{,}786{,}771 such commits across the 13-month window
of Claude availability, covering 185{,}517 distinct authors.

For each commit we record the SHA, the repository (owner/name), the commit
and authored timestamps, the committer and author identities (login,
email), the commit message, and the detected Claude model. We discard
1.6 million commits with missing \texttt{author\_login}, which are
predominantly bots and squash-merge artifacts.

\subsection{Sample construction}

We define a developer's \emph{adoption date} as her first Claude-co-authored
commit. This gives us 185{,}517 staggered adoption events distributed
across the 13-month window: 1{,}557 developers adopt in 2025\,Q1, 11{,}894
in Q2, 46{,}700 in Q3, 73{,}151 in Q4, and 52{,}215 in 2026\,Q1. The mass
of adoptions in late 2025 reflects the public launch of Claude Code in
May 2025 and the subsequent expansion of its free tier.

From this population we construct two analysis samples. The treated sample
contains 5{,}000 \emph{early adopters}: developers whose first Claude
commit falls in Q2 or Q3 of 2025. Within this window we further require at
least five Claude-co-authored commits to exclude one-time experimenters,
and we drop GitHub bot accounts identified by standard naming patterns
(e.g.\ \texttt{[bot]}, \texttt{dependabot}, \texttt{renovate}). The
eligible pool is 36{,}718 developers; we draw 5{,}000 stratified
proportionally by adoption month and equally across three intensity tiers
(5--10 commits, 11--50, 50+) to ensure variation in both treatment timing
and commitment. Because eligibility conditions on at least five
Claude-co-authored commits --- a post-adoption behavior --- the estimand
throughout is the event-time change among \emph{sustained adopters},
not among all developers who ever try the tool; one-time experimenters
are outside the sample by construction.

The control sample contains 5{,}000 \emph{not-yet-treated} developers:
authors who adopt Claude in 2025\,Q4 or 2026\,Q1 under the same eligibility
filters and stratification scheme. During the treated cohort's adoption
window of April--September 2025, the control developers have zero Claude
usage. This choice of comparison group is the standard recommendation in
the \citet{callaway2021difference} framework for staggered designs: the
controls are eventual adopters and therefore arguably similar in
unobserved tastes to the treated group, simply with a later adoption
date. Control developers serve as comparisons only in periods before
their own adoption.

Because name-pattern filters do not catch all automation, we additionally
screen the sampled accounts on behavioral criteria and exclude 13
single-repository automation accounts --- accounts committing more than
400 times per active month with over 90 percent of commits concentrated
in one repository (e.g.\ auto-generated content mirrors and data-logging
bots). These accounts contribute large commit volumes but essentially no
language information.

\paragraph{Competing agentic tools.}
Claude Code is not the only agentic assistant, and a developer's first
Claude commit is only informative about her first \emph{agentic}
experience if no other agent preceded it. We therefore screen every
commit message in the sample for the attribution signatures of other
agentic tools --- GitHub Copilot, OpenAI Codex, Cursor, Aider, Devin,
Gemini/Jules, and OpenHands --- using the same co-author-trailer
methodology that identifies Claude commits. Detectable competing-agent
use is an order of magnitude rarer than Claude use: the seven tools
combined co-author 0.8 percent of the sample's commits, against 21.7
percent for Claude, with GitHub Copilot the most common (8 percent of
developers show any use). We exclude the 214 developers (110 treated,
104 control) with any detectable competing-agent commit \emph{before
their own Claude adoption month}. The exclusion window is
developer-specific rather than a calendar date because each developer
plays two roles in the staggered design: the restriction guarantees
both that her pre-treatment history is agent-free and that she is
agent-free in every month she serves as a not-yet-treated comparison
for earlier cohorts --- both roles end at her own adoption.
Competing-agent use \emph{after} a developer's adoption is deliberately
not conditioned on: it is post-treatment behavior, plausibly caused by
the adoption event itself, and selecting on it would bias the
comparison. Detection is a lower bound --- tools that leave no commit
trace, such as autocomplete use of Copilot or Cursor, are invisible to
any commit-based measure --- a limitation we return to in Section
\ref{sec:discussion}.

\subsection{Commit-level contribution histories}
\label{sec:data-commitlevel}

GitHub attaches no language to a commit; language must be recovered from
the files each commit changed. Prior drafts of this literature (and an
earlier version of this paper) have relied on repository-level proxies
--- assigning every commit the \emph{primary language} of its repository
--- which mismeasures multi-language work inside repositories and
understates language breadth.\footnote{A repository hosting, say, a
TypeScript front end, a Python service, and SQL migrations is assigned
the single label ``TypeScript,'' so any work in the other languages is
invisible to a repository-level measure.} We therefore rebuild the
outcome data at the commit level in three steps.

\emph{Repository discovery.} For each sampled developer we assemble the
set of repositories she committed to during the window from three
sources: (i) GitHub's contribution API
(\texttt{contributionsCollection}); (ii) the repositories appearing in
the raw Claude-commit data; and (iii) the GH~Archive record of public
\texttt{PushEvent}s keyed to the developer's login. The third source is
important: the contribution API credits only commits on a repository's
default branch that are linked to the account's registered email, whereas
push events fire for any branch and any authoring email. Author-centric
archive discovery adds 45 percent more (developer, repository) pairs than
the first two sources combined, for a total of 133{,}096 pairs.

\emph{Commit enumeration and file extraction.} Within each discovered
repository we enumerate all commits authored by the sampled developer
with commit dates in the window, yielding 3{,}151{,}624 distinct commits
after de-duplicating commits that appear in both forks and upstream
repositories. For each commit we then extract the full list of changed
file paths and the complete commit message from bare metadata-only
clones of each repository, with a per-SHA API fallback for commits no
longer reachable from any branch (e.g.\ after force-pushes). The
procedure recovers changed-file records for 99.997 percent of enumerated
commits --- 3{,}151{,}536 commits touching 57.2 million files. Merge
commits are flagged by parent count and excluded from language
construction, since their file lists duplicate work already recorded in
the merged commits.

\emph{Language classification.} Each changed file is assigned a language
with GitHub's Linguist library, applied to the file path: exact-filename
rules (e.g.\ \texttt{Dockerfile}, \texttt{Makefile}) take precedence,
followed by extension rules, resolving ambiguous extensions to Linguist's
popular default. We retain only files Linguist classifies as
\emph{programming} languages --- excluding markup, data, and prose
formats (Markdown, JSON, YAML) as well as binaries --- and apply
Linguist's path-based vendored, generated, and documentation exclusions,
mirroring the rules GitHub itself uses for repository language
statistics.

\subsection{Outcome variables}
\label{sec:data-outcomes}

From the commit-level records we construct four \emph{primary} outcomes
that measure language-portfolio breadth, all defined over the
programming-language files changed by the developer's non-merge commits
in month $t$:

\begin{enumerate}
\item $N_{it}$, the number of distinct programming languages developer
$i$ works in during month $t$;
\item $\Delta N_{it}^{\text{new}}$, the number of languages first used
in month $t$ --- languages appearing in none of the developer's earlier
months in the observation window;\footnote{Novelty is measured relative
to the developer's observed history from January 2024 onward, not her
full lifetime history; a language used before 2024 and revived during
the window is counted as ``new.'' This measurement choice is symmetric
across treated and control developers and biases against finding
\emph{differential} novelty only if revival propensities differ by
group.}
\item $C_{it}$, the cumulative count of distinct languages used through
month $t$ (the portfolio stock); and
\item $H_{it}$, the Shannon entropy of the developer's file changes
across languages, $H_{it} = -\sum_{k} p_{itk}\log(p_{itk})$, where
$p_{itk}$ is the share of month-$t$ changed files in language $k$ ---
a secondary diversity measure capturing balance as well as count.
\end{enumerate}

The flow $\Delta N_{it}^{\text{new}}$ and the stock $C_{it}$ are the
most direct empirical counterparts of the model's activation-band and
portfolio-accumulation predictions.

By construction, these outcomes measure the languages of the
developer's \emph{shipped work under any production mode} --- solo,
augmented, or delegated --- and therefore include Claude-co-authored
commits. This is deliberate: the model's frontier object is the
production portfolio $N^2_{it}$, the set of languages in which the
developer can deliver working code by whichever mode she selects
(Remark \ref{rem:production-frontier}), not her unassisted coding
skill. Section \ref{sec:robustness} decomposes the outcomes into
assisted and unassisted components to make the distinction explicit.

Two further variables are reported as \emph{diagnostics} of engagement
and activity volume rather than as frontier outcomes: $R_{it}$, the
number of distinct repositories committed to in month $t$, and the
monthly commit count. The primary outcomes measure language-portfolio
breadth; the repository and commit outcomes characterize the scale of
activity surrounding adoption and feed the volume-based robustness
checks in Section \ref{sec:robustness}.

\subsection{Analysis sample and summary statistics}

The developer-month panel covers 5{,}825 developers over 28 months. For
estimation we impose two restrictions: developers must have at least one
commit in their pre-adoption window (excluding accounts with no
observable pre-treatment activity, for whom the Callaway--Sant'Anna
estimator places zero weight regardless), and no detectable
competing-agent use before their own adoption (the screen described
above). The estimation sample contains 5{,}346 developers (2{,}813
treated and 2{,}533 control) and 149{,}688 developer-month observations.
Table \ref{tab:summary} reports pre- and post-adoption means in the
treated group alongside the control mean, plus the simple cross-sectional
difference.

\begin{table}[H]
\caption{Summary statistics by group and period.}
\label{tab:summary}
\centering
\begin{threeparttable}
\begin{tabular}{lcccc}
\toprule
 & \multicolumn{2}{c}{Treated} & \multicolumn{1}{c}{Control} & \\
\cmidrule(lr){2-3}\cmidrule(lr){4-4}
Variable & Pre-adoption & Post-adoption & All months & Diff. (Post $-$ Control) \\
\midrule
Monthly commits & 10.68 & 40.11 & 15.91 & 24.20 \\
 & (55.89) & (113.06) & (75.83) & \\
Repositories & 0.94 & 2.03 & 0.95 & 1.08 \\
 & (2.67) & (3.71) & (2.26) & \\
Programming languages & 0.90 & 2.30 & 1.03 & 1.27 \\
 & (1.94) & (3.07) & (2.18) & \\
Language entropy (Shannon) & 0.15 & 0.37 & 0.16 & 0.20 \\
 & (0.33) & (0.48) & (0.36) & \\
Newly-used languages & 0.31 & 0.45 & 0.31 & 0.14 \\
 & (1.19) & (1.37) & (1.16) & \\
Cumulative languages & 3.06 & 8.88 & 3.71 & 5.17 \\
 & (4.49) & (6.39) & (5.08) & \\
\midrule
Developers & \multicolumn{2}{c}{2,813} & 2,533 & \\
Months & \multicolumn{4}{c}{28} \\
Observations (panel rows) & \multicolumn{4}{c}{149,688} \\
\bottomrule
\end{tabular}
\begin{tablenotes}\footnotesize
\item \emph{Notes.} Means with standard deviations in parentheses.
Outcomes are constructed from commit-level records: language variables
are defined over the programming-language files changed by the
developer's non-merge commits in the month (Section
\ref{sec:data-outcomes}). Means are taken over \emph{all}
developer-months in the balanced panel, including months with no
public commit activity (coded zero); GitHub activity is sporadic, so
means can fall below one --- conditional on using at least one
programming language in a month, the pre-adoption treated mean is 2.9
languages. ``Pre-adoption'' and ``Post-adoption'' refer to
event time relative to the treated developer's first Claude-co-authored
commit. The control column pools all 28 months for the later-adopting
cohorts; because these developers adopt within the panel (October
2025--January 2026), their final months are post-adoption --- in the
estimation they serve as comparisons only in months before their own
adoption. The fourth column is the cross-sectional difference between
treated-post and control means, not a causal effect. Sample restricted
to developers with pre-adoption commit activity and no detectable
competing-agent use before their own adoption.
\end{tablenotes}
\end{threeparttable}
\end{table}

The pattern in Table \ref{tab:summary} is consistent with substantial
within-developer change at adoption. Treated developers work in 1.4 more
programming languages per month post-adoption than pre-adoption (2.30
against 0.90, an increase of over 150 percent), and the differential
against the contemporaneous control mean (1.03) is similar in magnitude.
Monthly commits and repositories follow the same pattern. These are
uncontrolled comparisons that confound adoption with calendar-time
trends; the event-time estimates in Section \ref{sec:results} discipline
them through the staggered design.

\section{Empirical Strategy}
\label{sec:strategy}

The empirical task is to estimate event-time changes in developers'
language portfolios around first detectable Claude Code use. The
estimands of interest are the dynamic treatment-effect profiles for the
four primary language outcomes of Section \ref{sec:data-outcomes} ---
monthly distinct languages, newly-used languages, cumulative languages,
and language entropy --- with repositories and commits estimated
alongside as diagnostics. Because adoption is voluntary, we interpret
the estimates as event-time associations around adoption, disciplined by
the staggered design, rather than as definitive causal effects; Section
\ref{sec:discussion} discusses the selection threats in detail.

Two features of the setting dictate the choice of estimator. First,
treatment is \emph{staggered}: developers adopt Claude at different
months between April 2025 and January 2026. Second, treatment effects are
plausibly \emph{heterogeneous} across cohorts and time horizons: an
April 2025 adopter has had different exposure (to model upgrades, pricing
changes, and accumulated learning-by-doing) than a January 2026 adopter
at the same event time. Conventional two-way fixed-effects (TWFE)
specifications are known to be biased under these conditions
\citep{goodmanbacon2021difference,dechaisemartin2020two,sun2021estimating}.
We therefore adopt the group-time average treatment effect estimator of
\citet{callaway2021difference}.

\subsection{Group-time average treatment effects}

Let $i$ index developers, $t$ index months (coded $1, 2, \ldots, 28$ for
January 2024 through April 2026), and $G_i$ denote the period in which
developer $i$ first uses Claude detectably. All sampled developers
eventually adopt; at each date, cohorts that have not yet adopted serve
as the comparison group for cohorts that have, and each control developer
contributes comparisons only in periods before her own adoption. Let
$Y_{it}(0)$ denote the potential outcome for developer $i$ in period $t$
in the absence of treatment, and let $Y_{it}(g)$ denote the potential
outcome for a developer first treated in period $g$. The
\emph{group-time average treatment effect} is
\begin{equation}
ATT(g, t) = \E\!\left[\,Y_{it}(g) - Y_{it}(0)\,\middle|\,G_i = g\,\right],
\qquad t \geq g.
\label{eq:attgt}
\end{equation}
Each $ATT(g, t)$ is identified under three assumptions. Conditional
parallel trends requires
$\E[Y_{it}(0) - Y_{i,t-1}(0) \mid G_i = g] = \E[Y_{it}(0) - Y_{i,t-1}(0)
\mid G_i \in \mathcal{C}_t]$, where $\mathcal{C}_t$ is the set of
not-yet-treated comparison cohorts at period $t$. Limited treatment
anticipation bounds how far ahead of $g$ the outcome can respond; we
permit one month of anticipation, for reasons specific to how adoption
is measured (below). An overlap condition ensures that the propensity
score $P(G_i = g \mid X_i)$ is bounded away from zero and one on the
comparison support.

We estimate \eqref{eq:attgt} by the doubly robust (DR) estimator, which
combines a propensity-score weighting and an outcome regression and is
consistent if either is correctly specified
\citep{callaway2021difference}. The not-yet-treated control group widens
the comparison pool relative to the never-treated alternative and is
robust to compositional changes at later horizons. Standard errors are
computed by 1{,}000 multiplier-bootstrap iterations clustered at the
developer level, and uniform confidence bands are reported alongside
pointwise intervals throughout.

A remark on what the contrast identifies. Through the model's lens, the
comparison isolates the \emph{delegation margin}, not access to AI in
general: Generation-1 assistance was universally available throughout
the panel --- ChatGPT and Copilot predate the earliest adoption cohort
by more than two years (Section \ref{sec:background}) --- so
not-yet-treated developers face the same Generation-1 entry threshold
$T^1$, and the estimated contrast is the arrival of the delegation
threshold $T^D$. Symmetrically, comparison developers must not face
$T^D$ under another vendor's name, which is what the competing-agent
screen of Section \ref{sec:data} ensures.

\subsection{First \emph{detectable} use and anticipation}
\label{sec:strategy-anticipation}

The adoption date is the first commit carrying the Claude co-author
trailer --- the first \emph{detectable} use, not necessarily first
access. A developer typically experiments with the tool before her first
co-authored commit is pushed to a public repository: early sessions may
touch private repositories, produce commits without the trailer, or
involve no commits at all. The trailer-based date is therefore a
slightly \emph{late} measure of true adoption, and outcomes can begin
responding shortly before it. Following
\citet{callaway2021difference}, we accommodate this with an anticipation
window of one month: event time $e = -1$ is treated as potentially
affected, and the identifying parallel-trends restriction is imposed
(and tested) on event times $e \leq -2$. Interpreting a bump immediately
before measured adoption as anticipation rather than a trends violation
follows \citet{malani2015interpreting}; Section \ref{sec:results}
reports the corresponding diagnostics and shows the results are
unchanged with a two-month anticipation window.

\subsection{Aggregation: event-study and simple ATTs}

The matrix $\{ATT(g,t)\}$ is high-dimensional. We aggregate along two
dimensions following \citet{callaway2021difference}. The
\emph{event-study aggregation} averages $ATT(g, t)$ across cohorts at
each event time $e = t - g$:
\begin{equation}
ATT(e) = \sum_{g} \omega_g(e)\, ATT\!\left(g,\, g + e\right),
\qquad \omega_g(e) \geq 0,\ \textstyle\sum_g \omega_g(e) = 1,
\label{eq:event-study}
\end{equation}
with weights proportional to cohort size at event time $e$. This
collapses the $(g,t)$ matrix into a dynamic treatment-effect profile,
the central figure of the paper. The \emph{simple aggregation} averages
$ATT(g, t)$ across all post-treatment $(g, t)$ cells with weights
proportional to cohort-period sizes:
\begin{equation}
ATT^{\mathrm{simple}} = \sum_{g, t \geq g} \omega_{g,t} \, ATT(g, t).
\label{eq:simple-att}
\end{equation}
$ATT^{\mathrm{simple}}$ is interpretable as a single-number summary of
the average post-adoption change.

\subsection{Pre-trends}

Identification rests on the conditional parallel-trends assumption.
\citet{callaway2021difference} note that $ATT(g, t)$ for $t < g$ is
identified under the same DR machinery and should be zero in expectation
under the maintained assumption (outside the anticipation window). We
therefore report the event-study profile for $e \in \{-6,\ldots,-2\}$
alongside post-treatment values and test for pre-trend significance
directly. A non-zero pre-period profile at $e \leq -2$ would indicate
that treated and not-yet-treated cohorts followed different trajectories
prior to adoption, undermining the identifying assumption.

\subsection{Specification details}

For each outcome we estimate \eqref{eq:attgt} on the 5{,}346-developer
estimation sample described in Section \ref{sec:data}. The empirical
implementation uses the \texttt{csdid} package
\citep[d2cml-ai v0.4.2 implementation of][]{callaway2021difference} with
the following options: $\texttt{control\_group} = \texttt{notyettreated}$,
$\texttt{anticipation} = 1$, $\texttt{base\_period} = \texttt{varying}$,
$\texttt{est\_method} = \texttt{dr}$, and 1{,}000 multiplier-bootstrap
iterations. Event-study aggregation is reported for $e \in [-6, 10]$;
the simple aggregation pools all post-treatment $(g, t)$ cells.

Robustness checks targeting the mechanical linkage between the
treatment-defining commit and the language outcomes --- exclusion of the
first-Claude language, commit-volume controls, and per-commit rate
outcomes --- are reported in Section \ref{sec:robustness}.
Identification concerns --- the central threat being self-selection into
Claude adoption that coincides with the project shocks we are trying to
measure --- are taken up in Section \ref{sec:discussion} after the main
results are presented.

\section{Results: Programming-Language Portfolio Expansion}
\label{sec:results}

Around first detectable Claude Code use, developers' language portfolios
expand sharply. Table \ref{tab:main} summarizes the event-time estimates
for the four primary language outcomes and the two activity diagnostics;
Figures \ref{fig:event-langs} through \ref{fig:event-entropy} display the
corresponding event-study profiles. We present the language outcomes in
order of importance --- monthly distinct languages, newly-used languages,
cumulative languages, and entropy --- and then report repositories and
commits as supporting diagnostics.

\begin{table}[H]
\caption{Event-time estimates: ATT at event times 0, 1, 2 and simple
aggregation, by outcome.}
\label{tab:main}
\centering
\begin{threeparttable}
\begin{tabular}{lcccc}
\toprule
 & \multicolumn{3}{c}{Event-time ATT} & \\
\cmidrule(lr){2-4}
Outcome & $t=0$ & $t=1$ & $t=2$ & Simple ATT \\
\midrule
Monthly commits & $35.080^{*}$ & $21.046^{*}$ & $12.110^{*}$ & $18.677^{*}$ \\
 & $(2.085)$ & $(2.292)$ & $(2.609)$ & $(2.269)$ \\
Repositories & $1.494^{*}$ & $0.838^{*}$ & $0.449^{*}$ & $0.742^{*}$ \\
 & $(0.058)$ & $(0.057)$ & $(0.057)$ & $(0.057)$ \\
Programming languages & $2.528^{*}$ & $1.227^{*}$ & $0.693^{*}$ & $1.178^{*}$ \\
 & $(0.063)$ & $(0.064)$ & $(0.067)$ & $(0.052)$ \\
Language entropy & $0.382^{*}$ & $0.189^{*}$ & $0.102^{*}$ & $0.172^{*}$ \\
 & $(0.009)$ & $(0.011)$ & $(0.012)$ & $(0.010)$ \\
Newly-used languages & $1.193^{*}$ & $0.126^{*}$ & $-0.018$ & $0.294^{*}$ \\
 & $(0.051)$ & $(0.034)$ & $(0.038)$ & $(0.030)$ \\
Cumulative languages & $1.604^{*}$ & $1.892^{*}$ & $2.072^{*}$ & $1.930^{*}$ \\
 & $(0.054)$ & $(0.063)$ & $(0.086)$ & $(0.068)$ \\
\midrule
Developers & \multicolumn{4}{c}{5,346\;(\;2,813\;treated,\;2,533\;control)} \\
\bottomrule
\end{tabular}
\begin{tablenotes}\footnotesize
\item \emph{Notes.} Point estimates from the
\citet{callaway2021difference} doubly robust estimator with
not-yet-treated controls, one-month anticipation, and varying base
period. Outcomes are constructed from commit-level records; language
variables are defined over the programming-language files changed by
non-merge commits (Section \ref{sec:data-outcomes}). Bootstrap standard
errors in parentheses, clustered by developer (1{,}000 iterations).
$^{*}$ denotes $|t|>1.96$. ``Simple ATT'' is the cross-cohort,
cross-period average from Equation \eqref{eq:simple-att}.
\end{tablenotes}
\end{threeparttable}
\end{table}

\subsection{Monthly distinct languages}

The number of distinct programming languages a developer works in rises
by 2.53 (SE 0.06) in the adoption month, against a pre-adoption treated
mean of 0.90 --- the active language set roughly triples on impact. The
effect does not vanish with the adoption event: it remains at 1.23
(SE 0.06) one month later and 0.69 (SE 0.07) two months out, both
statistically significant, before decaying toward zero by the second
half of the first post-adoption year (Figure \ref{fig:event-langs});
point estimates at horizons beyond $e=4$ remain positive but are
imprecise, as fewer cohorts contribute. The simple aggregated ATT is
1.18 (SE 0.05) additional active languages per month. Pre-period
coefficients at $e \leq -2$ are tightly estimated zeros, consistent
with the identifying assumption.

\begin{figure}[H]
\caption{Event study: number of programming languages.}
\label{fig:event-langs}
\centering
\includegraphics[width=0.8\textwidth]{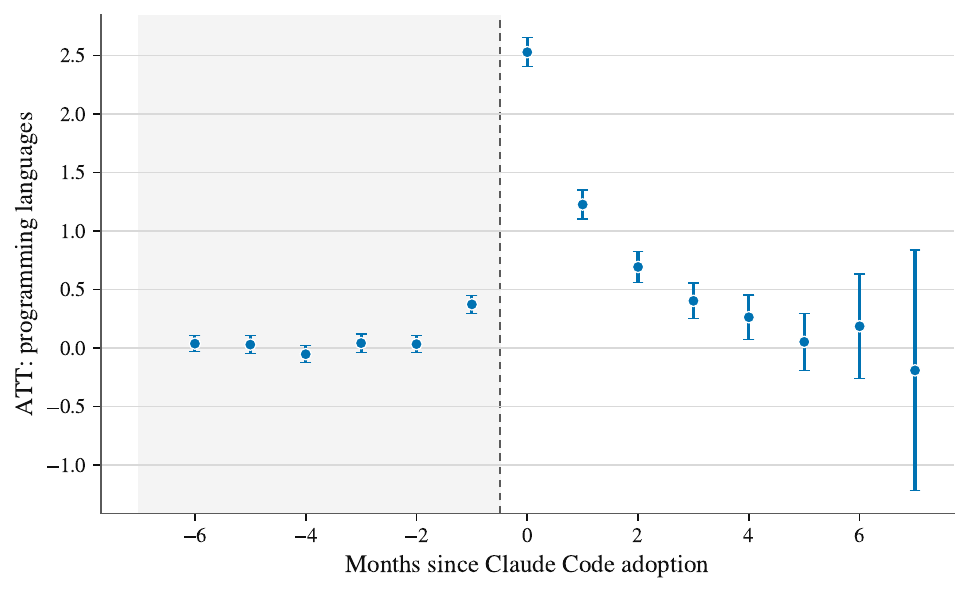}
\end{figure}

\subsection{Newly-used languages}

The most important outcome for the model is the flow of
\emph{newly-used} languages --- languages absent from the developer's
entire observed history. This is the empirical counterpart of the
activation band in Proposition \ref{prop:band}: delegation lowers the
entry threshold for some unfamiliar-language opportunities, and those
first uses should appear as a spike in the first-use flow at adoption.

That is what the data show (Figure \ref{fig:event-newlangs}). Developers
begin 1.19 (SE 0.05) new languages in the adoption month, nearly four
times the pre-adoption monthly flow of 0.31. Because
$\Delta N^{\text{new}}$ is a flow --- a given language can be ``new''
only once --- the profile necessarily reverts: the effect is 0.13
(SE 0.03) at $e=1$ and statistically indistinguishable from zero
thereafter. The economic content of the spike is therefore not
persistence of the flow itself but the entry events it records; whether
those entries accumulate is the question the stock measure answers next.
Part of the $e=0$ spike is mechanically linked to the
treatment-defining commit --- the first Claude-co-authored commit itself
often introduces a language --- and Section \ref{sec:robustness} shows
that roughly two-thirds of the spike survives excluding each developer's
first-Claude language entirely.

\begin{figure}[H]
\caption{Event study: newly-used languages.}
\label{fig:event-newlangs}
\centering
\includegraphics[width=0.8\textwidth]{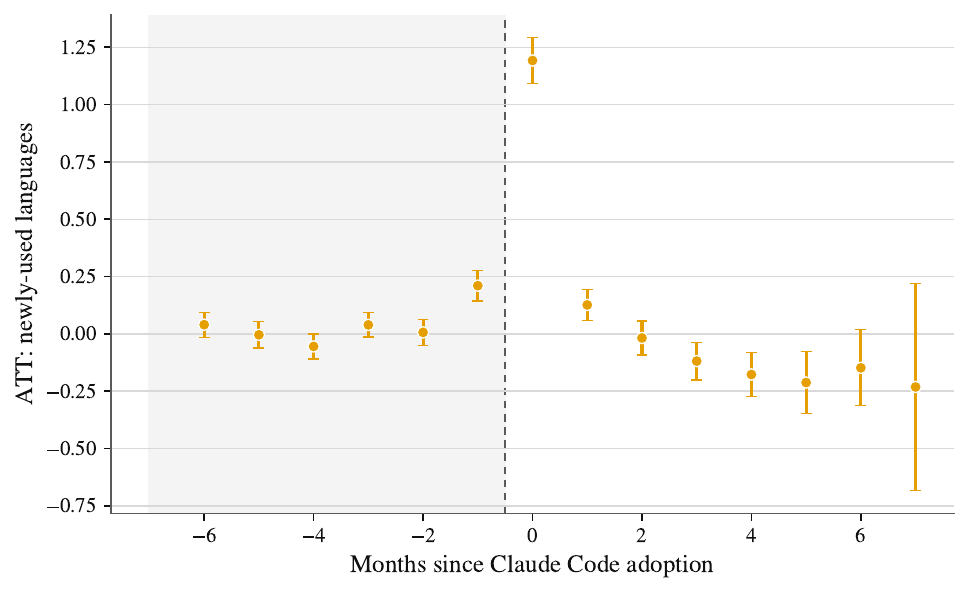}
\end{figure}

\subsection{Cumulative languages}

The cumulative-languages stock --- the running count of distinct
languages ever used --- rises by 1.60 (SE 0.05) at adoption and, unlike
every other outcome, \emph{grows} with event time: 1.89 at $e=1$, 2.07
at $e=2$, with a simple ATT of 1.93 (Figure \ref{fig:event-cumlangs}).
This is the stock dynamic anticipated by Proposition \ref{prop:dynamic}:
first-use events accumulate, so the stock can keep rising even as the
newly-used-language flow reverts.

An important caveat applies. Four of the five pre-period coefficients
for this outcome are statistically significant, indicating that treated
developers were already accumulating languages somewhat faster than
not-yet-treated developers before adoption. Unlike the monthly-flow
outcomes, the cumulative stock therefore does not satisfy the
parallel-trends diagnostic cleanly, and we present its profile as
descriptive evidence on stock dynamics rather than as a headline
estimate. The monthly language count and the newly-used-language flow
--- both with clean pre-periods --- carry the identification weight.

\begin{figure}[H]
\caption{Event study: cumulative languages ever used.}
\label{fig:event-cumlangs}
\centering
\includegraphics[width=0.8\textwidth]{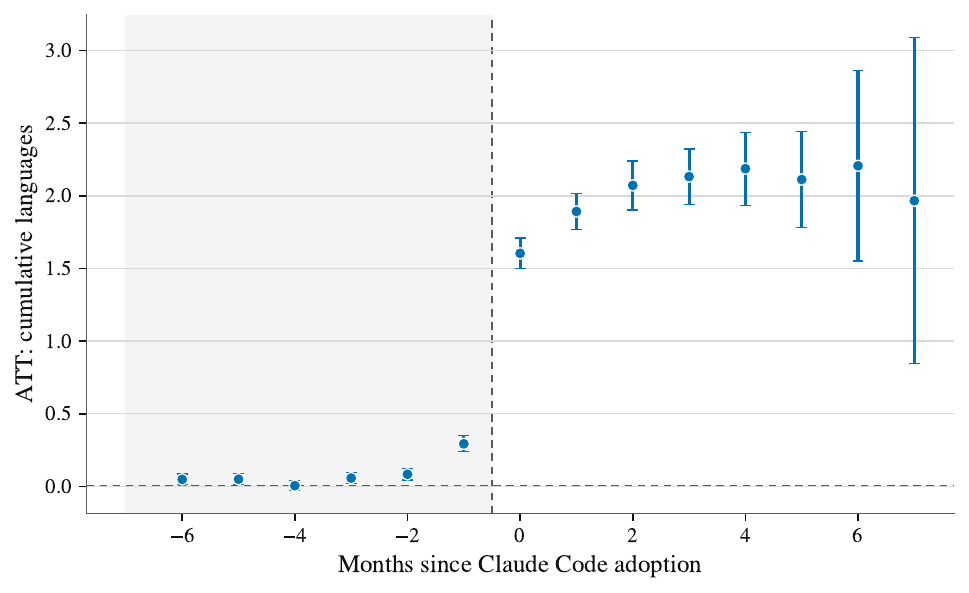}
\end{figure}

\subsection{Language entropy}

Shannon entropy over the month's changed files weights both the count
and the balance of the language mix: it rises only if new languages
carry non-trivial shares of the developer's work. Entropy increases by
0.38 (SE 0.01) at adoption --- against a pre-adoption mean of 0.15 ---
and remains significantly elevated at $e=1$ (0.19) and $e=2$ (0.10),
with a simple ATT of 0.17 (Figure \ref{fig:event-entropy}). The entropy
result indicates that the added languages are not single-file touches:
the distribution of work across languages genuinely broadens. We treat
entropy as a secondary diversity measure, since the model speaks to
which languages enter the portfolio rather than to allocation within
it.

\begin{figure}[H]
\caption{Event study: language entropy (Shannon).}
\label{fig:event-entropy}
\centering
\includegraphics[width=0.8\textwidth]{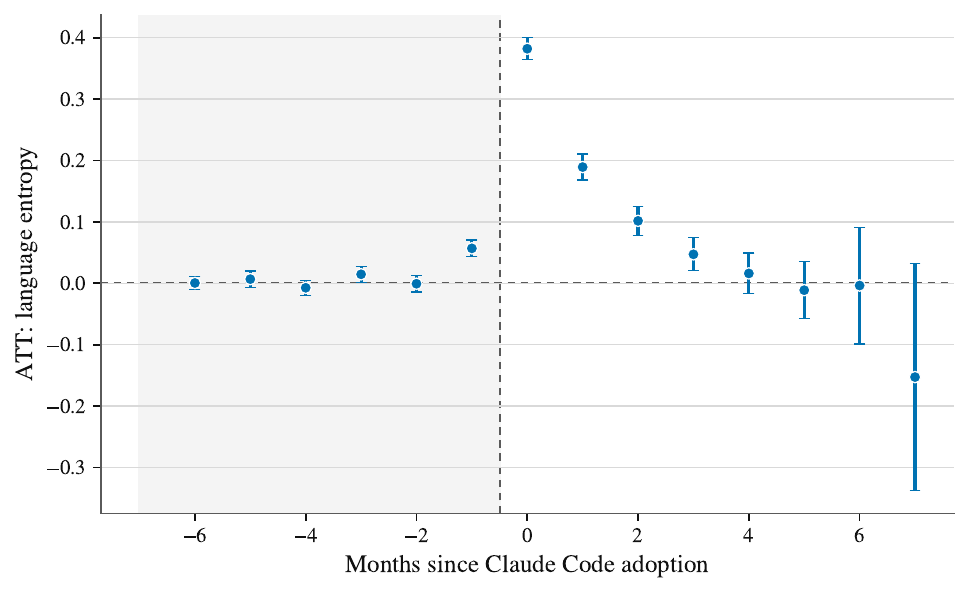}
\end{figure}

\subsection{Diagnostics: repositories and commits}

Repository and commit counts are activity measures, not frontier
measures; we report them to characterize the adoption event and because
they discipline the interpretation of the language results. Treated
developers contribute to 1.49 more distinct repositories (SE 0.06) in
the adoption month --- more than a doubling relative to the
pre-adoption mean of 0.94 --- with a simple ATT of 0.74 (Figure
\ref{fig:event-repos}). Monthly commits rise by 35.1 (SE 2.1) at
adoption against a pre-adoption mean of 10.7, with a simple ATT of 18.7.

The commit result matters for interpretation: adoption is accompanied by
a large increase in the \emph{volume} of activity, so part of any
language-count increase could reflect more commits mechanically touching
more languages. Section \ref{sec:robustness} addresses this directly ---
conditioning on pre-adoption activity leaves the language estimates
unchanged, while per-commit rate outcomes show that the adoption-month
jump includes genuine per-commit diversification but the persistent
component operates through the expanded volume of work spanning more
languages.

\begin{figure}[H]
\caption{Event study: number of repositories (diagnostic).}
\label{fig:event-repos}
\centering
\includegraphics[width=0.8\textwidth]{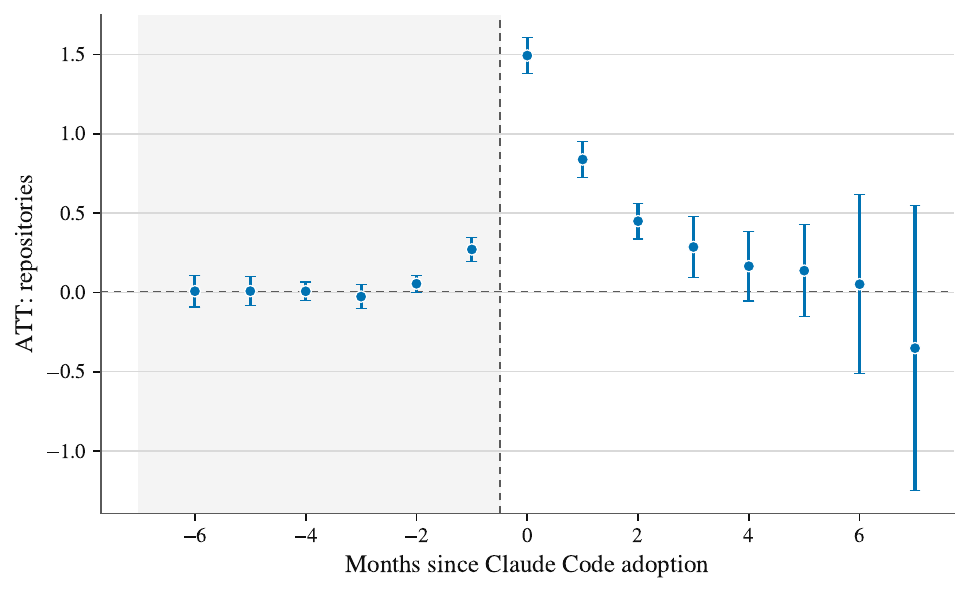}
\end{figure}

\subsection{Pre-trends and anticipation}
\label{sec:results-pretrends}

For the monthly language count, the newly-used-language flow, and
entropy, the pre-period coefficients at $e \in \{-6, \ldots, -2\}$ are
small, statistically insignificant, and exhibit no trend --- visible in
the flat shaded regions of Figures
\ref{fig:event-langs}--\ref{fig:event-entropy} --- supporting the
parallel-trends assumption where it is imposed. At $e = -1$ the language
count shows a positive coefficient of 0.37: the single month inside the
anticipation window is partially treated, exactly as the
first-\emph{detectable}-use measurement predicts (Section
\ref{sec:strategy-anticipation}). Two facts support the anticipation
reading over a trends violation: the bump is confined to $e=-1$ with no
build-up over $e \leq -2$, and it is essentially unchanged when we widen
the anticipation window to two months or re-date adoption using each
developer's earliest Claude commit recovered from the complete
commit-level history. Extending the pre-period window to a full year
leaves the profile flat --- all eleven coefficients over
$e \in \{-12,\ldots,-2\}$ are tightly estimated zeros, with every cohort
contributing at every point (Appendix \ref{app:support}, Figure
\ref{fig:extended-pretrend}); Appendix Table \ref{tab:risksets} reports
the cohorts and developers contributing at each event time. The exception to clean pre-periods is the
cumulative stock, discussed above, which we accordingly do not treat as
identified evidence.

\subsection{Magnitudes in economic terms}

Pooling the results, a representative treated developer in the adoption
month works in roughly 2.5 additional programming languages, begins
about 1.2 languages she has never used in her observed history,
contributes to 1.5 additional repositories, and makes about 35
additional commits. Averaged across the post-adoption window, the
portfolio expansion settles at 1.2 additional active languages per month
--- more than doubling the pre-adoption baseline of 0.90 --- and the
newly-used-language flow averages 0.29 per month above baseline. These
are large movements of the language frontier for individual developers;
whether they reflect the causal effect of agentic delegation or
project-driven selection into adoption is taken up in Sections
\ref{sec:robustness} and \ref{sec:discussion}.

\section{Robustness: Mechanical Exposure and Activity Volume}
\label{sec:robustness}

The central threat to interpreting Section \ref{sec:results} is that
treatment and outcome are \emph{mechanically} linked: adoption is dated
by a commit, and commits carry languages. The adoption-month estimate
could in principle reflect nothing more than the treatment-defining
commit itself introducing a language, amplified by the surge in commit
volume that accompanies adoption. This section addresses the mechanical
channel directly with four families of checks: excluding the language
introduced by the first Claude commit, excluding \emph{all}
Claude-co-authored commits from outcome construction, controlling for
activity volume, and re-estimating under stricter sample and timing
assumptions.

\subsection{Excluding the first-Claude language}
\label{sec:robust-purge}

The mechanical channel is real and quantitatively important to rule
out: 77 percent of first Claude-co-authored commits touch at least one
programming language (1.78 languages on average), so the treatment
event directly deposits languages into the adoption-month outcome. We
therefore rebuild every language outcome after deleting, for each
developer, \emph{all} activity in the language(s) her first Claude
commit introduced --- not only at $e=0$ but in every month, so the
purged outcomes cannot register that language at any horizon. This
removes 7{,}529 (developer, language) pairs.

\begin{table}[H]
\caption{Mechanical-exposure robustness: baseline vs.\ first-Claude
language excluded.}
\label{tab:purge}
\centering
\begin{threeparttable}
\begin{tabular}{lcccccc}
\toprule
 & \multicolumn{3}{c}{Baseline} & \multicolumn{3}{c}{First-Claude language excluded} \\
\cmidrule(lr){2-4}\cmidrule(lr){5-7}
Outcome & $t=0$ & $t=1$ & $t=2$ & $t=0$ & $t=1$ & $t=2$ \\
\midrule
Programming languages & $2.528^{*}$ & $1.227^{*}$ & $0.693^{*}$ & $1.580^{*}$ & $0.898^{*}$ & $0.545^{*}$ \\
 & $(0.063)$ & $(0.064)$ & $(0.067)$ & $(0.055)$ & $(0.053)$ & $(0.060)$ \\
Language entropy & $0.382^{*}$ & $0.189^{*}$ & $0.102^{*}$ & $0.249^{*}$ & $0.138^{*}$ & $0.084^{*}$ \\
 & $(0.009)$ & $(0.011)$ & $(0.012)$ & $(0.009)$ & $(0.010)$ & $(0.010)$ \\
Newly-used languages & $1.193^{*}$ & $0.126^{*}$ & $-0.018$ & $0.807^{*}$ & $0.184^{*}$ & $0.053$ \\
 & $(0.051)$ & $(0.034)$ & $(0.038)$ & $(0.041)$ & $(0.032)$ & $(0.036)$ \\
\bottomrule
\end{tabular}
\begin{tablenotes}\footnotesize
\item \emph{Notes.} Each panel reports event-study ATTs from a separate
\citet{callaway2021difference} estimation. In the right panel, all
activity in the language(s) contained in each developer's first
Claude-co-authored commit is removed from outcome construction in every
month. Bootstrap standard errors in parentheses, clustered by developer
(1{,}000 iterations). $^{*}$ denotes $|t|>1.96$.
\end{tablenotes}
\end{threeparttable}
\end{table}

Table \ref{tab:purge} shows the estimates survive. The monthly language
count still rises by 1.58 (SE 0.06) at adoption, remains at 0.90 one
month out and 0.55 two months out, all statistically significant. The
newly-used-language spike falls from 1.19 to 0.81 --- so roughly
one-third of the adoption-month first-use flow is attributable to the
treatment-defining language and two-thirds is expansion \emph{beyond}
it. Entropy behaves identically. Figure \ref{fig:purge} overlays the
baseline and purged event-study profiles for the language count: the
two paths differ at $e \in \{0, 1\}$ --- the mechanical component ---
and converge from $e=2$ onward. The sustained portion of the effect is
therefore not the mechanically introduced language.

\begin{figure}[H]
\caption{Event study: languages, baseline vs.\ first-Claude language
excluded.}
\label{fig:purge}
\centering
\includegraphics[width=0.8\textwidth]{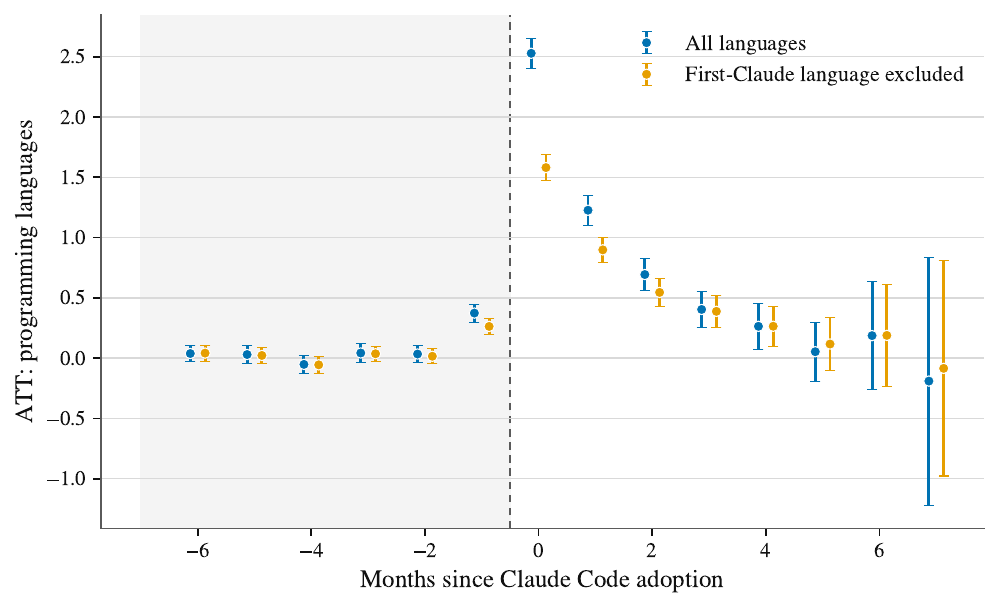}
\end{figure}

\subsection{Decomposing assisted and unassisted production}
\label{sec:robust-decomp}

The strictest version of the mechanical-exposure concern is that the
outcomes simply count the agent's own output: delete the
Claude-co-authored commits, on this reading, and the developer's
portfolio is unchanged. The model takes no offense at the premise ---
its frontier object is the production portfolio, delegated work
included (Remark \ref{rem:production-frontier}) --- but the decomposition
is informative about mechanism. We therefore rebuild every outcome using
\emph{only} commits without the Claude co-author trailer: the
``unassisted portfolio.'' Treatment timing, sample, and estimator are
unchanged; the treatment-defining commits and all other assisted work
are simply removed from outcome construction.

\begin{table}[H]
\caption{Decomposition: all production vs.\ unassisted production only.}
\label{tab:decomp}
\centering
\begin{threeparttable}
\begin{tabular}{lcccccc}
\toprule
 & \multicolumn{3}{c}{All production} & \multicolumn{3}{c}{Unassisted production only} \\
\cmidrule(lr){2-4}\cmidrule(lr){5-7}
Outcome & $t=0$ & $t=1$ & $t=2$ & $t=0$ & $t=1$ & $t=2$ \\
\midrule
Programming languages & $2.528^{*}$ & $1.227^{*}$ & $0.693^{*}$ & $1.663^{*}$ & $0.685^{*}$ & $0.285^{*}$ \\
 & $(0.063)$ & $(0.064)$ & $(0.067)$ & $(0.061)$ & $(0.058)$ & $(0.059)$ \\
Language entropy & $0.382^{*}$ & $0.189^{*}$ & $0.102^{*}$ & $0.268^{*}$ & $0.120^{*}$ & $0.049^{*}$ \\
 & $(0.009)$ & $(0.011)$ & $(0.012)$ & $(0.009)$ & $(0.010)$ & $(0.011)$ \\
Newly-used languages & $1.193^{*}$ & $0.126^{*}$ & $-0.018$ & $0.724^{*}$ & $0.054$ & $-0.077^{*}$ \\
 & $(0.051)$ & $(0.034)$ & $(0.038)$ & $(0.042)$ & $(0.036)$ & $(0.035)$ \\
Cumulative languages & $1.604^{*}$ & $1.892^{*}$ & $2.072^{*}$ & $1.130^{*}$ & $1.331^{*}$ & $1.434^{*}$ \\
 & $(0.054)$ & $(0.063)$ & $(0.086)$ & $(0.048)$ & $(0.058)$ & $(0.072)$ \\
\bottomrule
\end{tabular}
\begin{tablenotes}\footnotesize
\item \emph{Notes.} The left panel reproduces the main estimates
(outcomes over all commits). In the right panel, every outcome is
rebuilt from commits \emph{without} the Claude co-author trailer, so no
Claude-co-authored work enters outcome construction at any horizon.
Treatment timing, sample (5{,}346 developers), and estimator are
identical across panels. Bootstrap standard errors in parentheses,
clustered by developer (1{,}000 iterations). $^{*}$ denotes $|t|>1.96$.
\end{tablenotes}
\end{threeparttable}
\end{table}

Table \ref{tab:decomp} shows the expansion is not confined to
agent-co-authored output. Restricting to unassisted commits, the monthly
language count still rises by 1.66 (SE 0.06) at adoption and remains
significantly elevated at $e=1$ (0.68) and $e=2$ (0.28); entropy and the
newly-used-language flow behave in parallel (0.27 and 0.72 at $e=0$,
both significant). Unassisted commit volume itself rises by 18.3 at
adoption --- developers do more of their own work, in more languages,
around adoption. Because a commit without the trailer is not guaranteed
to be free of AI assistance (trailers can be stripped by history
rewrites), the unassisted panel should be read as an approximation; but
the pattern is consistent with developers personally engaging with the
languages of newly entered projects --- editing, fixing, and extending
code the agent produced --- and, at later horizons, with the learning
extension of Appendix \ref{app:proofs}. The decomposition rules out the
sharpest deflationary reading: the portfolio expansion is not an
artifact of counting machine output, since it survives with every
Claude-co-authored commit deleted.

\subsection{Activity volume}
\label{sec:robust-volume}

Adoption raises monthly commits by roughly 35 at $e=0$; more commits
could mechanically touch more languages. We probe the volume channel
from two directions.

First, we condition the doubly robust estimator on each developer's
pre-adoption activity --- log pre-period commits and log pre-period
repositories --- so that treated developers are compared to
not-yet-treated developers with similar baseline volume. Table
\ref{tab:volume} (top panel) shows the language-count profile is
unchanged to the second decimal (2.53 versus 2.53 at $e=0$; 1.23 versus
1.23 at $e=1$), and Figure \ref{fig:volumectrl} shows the two profiles
lie on top of each other at every event time. The estimates are not
driven by selection of high-volume developers into adoption.

\begin{table}[H]
\caption{Volume robustness: baseline-activity controls and per-commit
rates.}
\label{tab:volume}
\centering
\begin{threeparttable}
\begin{tabular}{lccc}
\toprule
Specification & $t=0$ & $t=1$ & $t=2$ \\
\midrule
\multicolumn{4}{l}{\textit{Programming languages}} \\
\quad No controls & $2.528^{*}$ & $1.227^{*}$ & $0.693^{*}$ \\
\quad Baseline volume + repos & $2.530^{*}$ & $1.228^{*}$ & $0.695^{*}$ \\
\multicolumn{4}{l}{\textit{Per-commit rates}} \\
\quad Languages per commit & $0.0546^{*}$ & $0.0045$ & $-0.0092$ \\
\quad New languages per commit & $0.0482^{*}$ & $-0.0257^{*}$ & $-0.0240^{*}$ \\
\bottomrule
\end{tabular}
\begin{tablenotes}\footnotesize
\item \emph{Notes.} Top panel: ATTs for the monthly language count with
and without conditioning the doubly robust estimator on log pre-adoption
commits and repositories. Bottom panel: per-commit rate outcomes
(outcome divided by the month's commit count; zero-commit months coded
zero). Bootstrap standard errors clustered by developer. $^{*}$ denotes
$|t|>1.96$.
\end{tablenotes}
\end{threeparttable}
\end{table}

\begin{figure}[H]
\caption{Event study: languages, with vs.\ without baseline-volume
controls.}
\label{fig:volumectrl}
\centering
\includegraphics[width=0.8\textwidth]{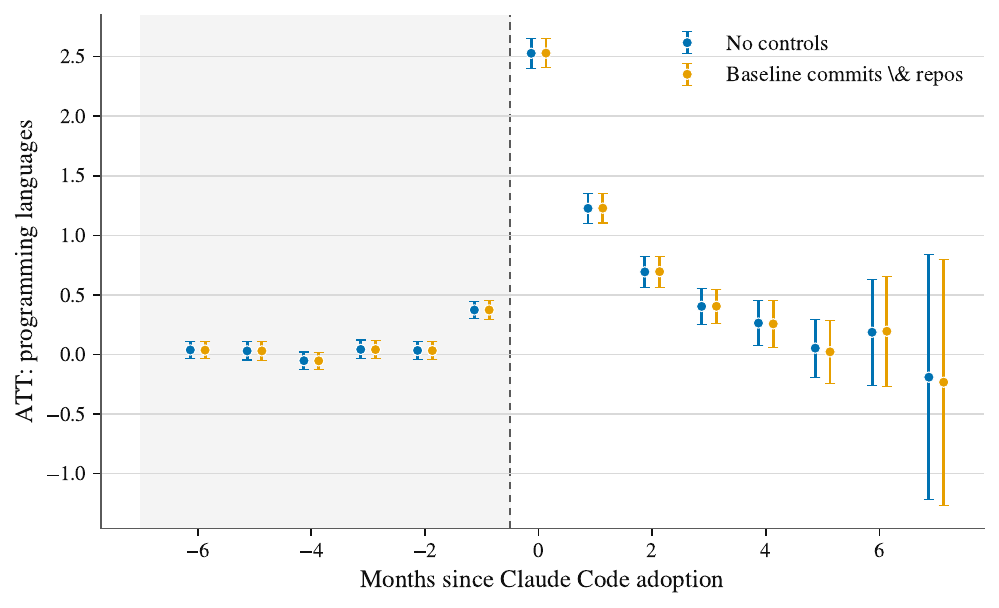}
\end{figure}

Enriching the covariate set with pre-period language count and language
entropy alongside the volume measures again moves nothing (2.54 at
$e=0$ for the language count, against 2.53 unconditioned).

Second, we net out contemporaneous volume entirely by estimating
per-commit \emph{rate} outcomes (Table \ref{tab:volume}, bottom panel;
Figure \ref{fig:rate}). Languages per commit rise by 0.055 (SE 0.013)
in the adoption month --- commits genuinely become more
language-diverse on impact --- but the rate effect disappears at $e=1$
and turns slightly negative thereafter. The two probes together give a
precise decomposition of the mechanism: the adoption month features
per-commit diversification over and above the volume surge, while the
\emph{persistent} elevation in the monthly language count operates
through the expanded volume of work spanning more languages, not
through each commit individually becoming more polyglot. We read this
as economically sensible --- delegation enables more building across a
wider set of languages --- and we are careful in Section
\ref{sec:discussion} not to claim persistent per-commit
diversification.

\begin{figure}[H]
\caption{Event study: languages per commit.}
\label{fig:rate}
\centering
\includegraphics[width=0.8\textwidth]{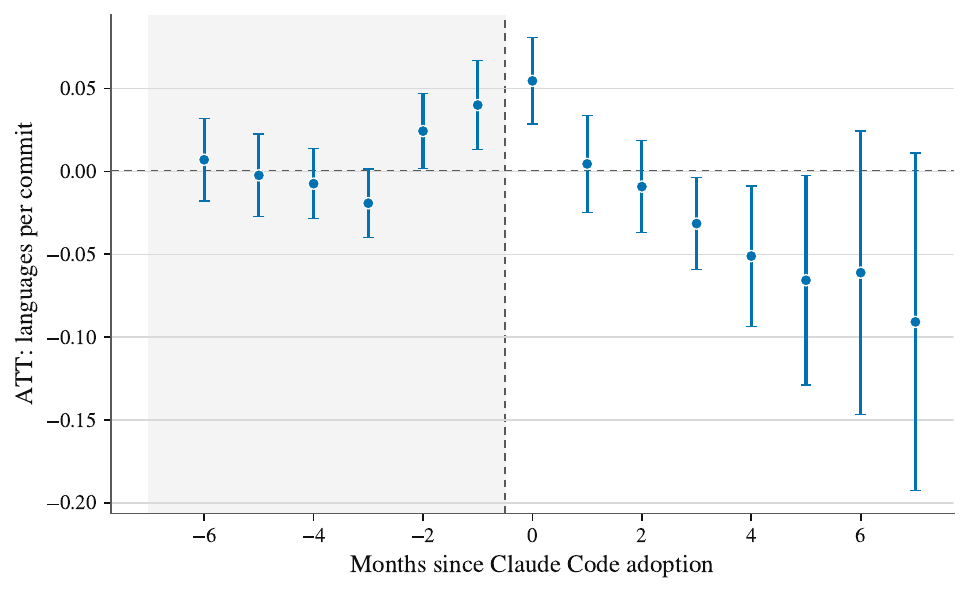}
\end{figure}

\subsection{Stricter pre-period activity filters}
\label{sec:robust-filters}

The main sample admits developers with as little as one pre-adoption
commit. To verify that low-pre-activity developers do not drive the
estimates through level shifts from a near-zero baseline, we re-estimate
under two stricter filters: developers active in at least 50 percent of
their pre-treatment months, and developers active in at least six
pre-treatment months.

\begin{table}[H]
\caption{Robustness: ATT at $t=0$ across sample restrictions.}
\label{tab:robust}
\centering
\begin{threeparttable}
\begin{tabular}{lccc}
\toprule
 & \multicolumn{3}{c}{ATT at adoption ($t=0$)} \\
\cmidrule(lr){2-4}
Outcome & Main & Active $\geq 50\%$ & Active $\geq 6$ months \\
\midrule
Programming languages & $2.528^{*}$ & $2.087^{*}$ & $2.183^{*}$ \\
 & $(0.063)$ & $(0.117)$ & $(0.096)$ \\
Language entropy & $0.382^{*}$ & $0.280^{*}$ & $0.324^{*}$ \\
 & $(0.009)$ & $(0.020)$ & $(0.016)$ \\
Newly-used languages & $1.193^{*}$ & $0.675^{*}$ & $0.757^{*}$ \\
 & $(0.051)$ & $(0.069)$ & $(0.055)$ \\
Cumulative languages & $1.604^{*}$ & $0.683^{*}$ & $0.757^{*}$ \\
 & $(0.054)$ & $(0.076)$ & $(0.062)$ \\
Repositories & $1.494^{*}$ & $1.622^{*}$ & $1.572^{*}$ \\
 & $(0.058)$ & $(0.157)$ & $(0.104)$ \\
Monthly commits & $35.080^{*}$ & $43.770^{*}$ & $36.528^{*}$ \\
 & $(2.085)$ & $(3.475)$ & $(4.109)$ \\
\midrule
Developers & 5,346 & 1,425 & 2,393 \\
\bottomrule
\end{tabular}
\begin{tablenotes}\footnotesize
\item \emph{Notes.} Each cell reports the ATT at event time 0 from a
separate \citet{callaway2021difference} estimation under the indicated
sample restriction. ``Active $\geq 50\%$'' retains developers with
commits in at least half of their pre-treatment months; ``Active $\geq
6$ months'' retains developers with commits in at least six
pre-treatment months. Bootstrap standard errors in parentheses,
clustered by developer. $^{*}$ denotes $|t|>1.96$.
\end{tablenotes}
\end{threeparttable}
\end{table}

Table \ref{tab:robust} shows signs and significance hold uniformly
across the restricted samples. The monthly language count attenuates
only modestly (2.53 in the main specification versus 2.09 and 2.18
under the two filters), and the repository and commit diagnostics are
essentially unchanged. The newly-used-language flow attenuates more
(1.19 versus 0.68 and 0.76), which is expected rather than concerning:
consistently active developers have denser observed histories, so fewer
languages remain that can count as first uses within the window.
Developers with established pre-adoption activity patterns expand their
language portfolios at adoption much as the full sample does. Relaxing
the sample in the other direction --- re-including the 214 developers
excluded by the competing-agent screen of Section \ref{sec:data} ---
yields an adoption-month language ATT of 2.46 (SE 0.06), so the screen
sharpens the estimand but does not drive the estimates.

\subsection{Timing assumptions}
\label{sec:robust-timing}

Finally, the results are insensitive to how the anticipation window and
the adoption date are set. Widening the anticipation window to two
months --- so that the parallel-trends restriction is imposed and tested
on $e \leq -3$ --- leaves the language-count profile essentially
unchanged (ATT at $e=0$ of 2.62 versus 2.53) with a flat pre-period.
Re-dating each developer's adoption to the earliest Claude-co-authored
commit recovered in the complete commit-level history --- which shifts
154 developers to earlier adoption dates than the original
event-stream measure --- likewise yields a nearly identical profile
(ATT at $e=0$ of 2.77). Both checks reinforce the reading of the
$e=-1$ coefficient as detection lag rather than differential
pre-trends (Section \ref{sec:results-pretrends}).

\subsection{Placebo adoption dates}
\label{sec:robust-placebo}

Finally, a falsification test. We assign every developer a fake
adoption date twelve months before her real one, restrict the panel to
months strictly before real adoption --- so that no genuinely treated
observation can contaminate the exercise --- and re-run the estimator.
If the design were picking up differential trends or spurious timing
patterns, the placebo would find them. It finds nothing: the effect at
the fake adoption month is 0.003 for the language count, $-0.019$ for
newly-used languages, and 0.005 for entropy, and not a single
post-placebo coefficient is statistically significant for any of the
three outcomes (zero of seven event times each). The event-time
response is specific to the true adoption date.

\subsection{Summary}

The language-portfolio expansion is not an artifact of the
treatment-defining commit, of counting agent-co-authored output, of
activity volume, of low-activity developers, or of timing assumptions,
and it does not appear at placebo dates. Because Tables
\ref{tab:main}--\ref{tab:decomp} report the $t=1$ and $t=2$
coefficients throughout, inference that excludes the potentially
mechanical adoption month entirely can be read directly off those
columns.
The checks also sharpen the interpretation: about one-third of the
adoption-month first-use spike is the first Claude commit's own
language; roughly two-thirds of the language-count effect appears even
in commits without any Claude co-authorship; and the persistence
reflects an expanded volume of work --- assisted and unassisted ---
spanning more languages.

\section{Mechanism: Specialist and Ability Heterogeneity}
\label{sec:mechanism}

The model's sharpest cross-sectional prediction (Appendix
\ref{app:proofs}, Proposition \ref{prop:specialist}) is that expected
expansion into unfamiliar languages equals $U_i \, p_i(a_i)$: it is
increasing in the stock of unfamiliar-language candidates $U_i$ ---
\emph{specialists}, with narrow pre-adoption portfolios, have the most
headroom --- and in general ability $a_i$, which lowers verification
costs. This heterogeneity test matters beyond the model: a generic
activity shock at adoption would raise output for everyone roughly
proportionally, whereas the delegation mechanism predicts that
\emph{first uses} concentrate among developers with many unfamiliar
candidates.

\subsection{Design}

We classify developers by two pre-treatment characteristics, each
measured over the developer's own pre-adoption window. \emph{Ability}
is proxied by pre-adoption commit volume; \emph{breadth} is the count
of distinct programming languages in pre-adoption commits. Observed
breadth is mechanically contaminated by volume --- a developer with five
pre-adoption commits looks narrow simply because we barely observe her
--- so we double-sort: developers are first split at the median of
pre-adoption volume, and \emph{specialist} is then defined as
below-median breadth \emph{within} each ability half. This breaks the
mechanical correlation and yields four cells of 1{,}174--1{,}494
developers. We estimate the full Callaway--Sant'Anna specification
separately within each cell, so specialists are compared only to
not-yet-treated specialists of the same ability level.

\begin{table}[H]
\caption{Heterogeneity: ATTs by pre-adoption ability and breadth.}
\label{tab:specialist}
\centering
\begin{threeparttable}
\begin{tabular}{lcccc}
\toprule
 & \multicolumn{2}{c}{High ability} & \multicolumn{2}{c}{Low ability} \\
\cmidrule(lr){2-3}\cmidrule(lr){4-5}
 & Specialist & Generalist & Specialist & Generalist \\
\midrule
\multicolumn{5}{l}{\textit{Programming languages}} \\
\quad $t=0$ & $2.218^{*}$ & $2.149^{*}$ & $2.864^{*}$ & $2.885^{*}$ \\
 & $(0.107)$ & $(0.173)$ & $(0.095)$ & $(0.116)$ \\
\quad $t=1$ & $1.068^{*}$ & $1.025^{*}$ & $1.547^{*}$ & $1.254^{*}$ \\
 & $(0.102)$ & $(0.183)$ & $(0.092)$ & $(0.128)$ \\
\quad $t=2$ & $0.628^{*}$ & $0.176$ & $1.075^{*}$ & $0.857^{*}$ \\
 & $(0.106)$ & $(0.183)$ & $(0.091)$ & $(0.149)$ \\
\quad Simple ATT & $1.027^{*}$ & $0.735^{*}$ & $1.580^{*}$ & $1.355^{*}$ \\
 & $(0.087)$ & $(0.169)$ & $(0.063)$ & $(0.112)$ \\
\multicolumn{5}{l}{\textit{Newly-used languages}} \\
\quad $t=0$ & $0.981^{*}$ & $0.301^{*}$ & $2.388^{*}$ & $1.015^{*}$ \\
 & $(0.084)$ & $(0.100)$ & $(0.099)$ & $(0.092)$ \\
\quad $t=1$ & $0.137^{*}$ & $-0.046$ & $0.410^{*}$ & $-0.019$ \\
 & $(0.047)$ & $(0.099)$ & $(0.054)$ & $(0.083)$ \\
\quad $t=2$ & $0.021$ & $-0.306^{*}$ & $0.282^{*}$ & $-0.102$ \\
 & $(0.055)$ & $(0.095)$ & $(0.050)$ & $(0.103)$ \\
\quad Simple ATT & $0.266^{*}$ & $-0.100$ & $0.846^{*}$ & $0.149^{*}$ \\
 & $(0.044)$ & $(0.088)$ & $(0.036)$ & $(0.074)$ \\
\midrule
Pre-adoption languages & 0.75 & 2.11 & 0.07 & 0.40 \\
Developers & 1,494 & 1,174 & 1,408 & 1,270 \\
\bottomrule
\end{tabular}
\begin{tablenotes}\footnotesize
\item \emph{Notes.} Each column reports a separate
\citet{callaway2021difference} estimation on the indicated cell of the
double sort (median pre-adoption commit volume, then median pre-adoption
language breadth within ability halves). Bootstrap standard errors in
parentheses, clustered by developer (1{,}000 iterations). $^{*}$
denotes $|t|>1.96$. ``Pre-adoption languages'' is the cell's mean
monthly language count over pre-adoption months.
\end{tablenotes}
\end{threeparttable}
\end{table}

\subsection{Results}

Table \ref{tab:specialist} and Figure \ref{fig:specialist} show two
distinct patterns.

\emph{The headroom margin works as predicted.} On newly-used languages
--- the outcome that maps directly to the activation band ---
specialists expand far more than generalists \emph{within both ability
halves}: 0.98 versus 0.30 at $e=0$ among high-ability developers (a
threefold gap), 2.39 versus 1.02 among low-ability developers, with the
same ordering in the simple ATTs (0.27 versus $-0.10$; 0.85 versus
0.15). Because the comparison is within ability halves, this gap cannot
be an artifact of activity volume. The monthly language \emph{count},
by contrast, jumps by similar level amounts across cells (2.1--2.9 at
$e=0$) --- but specialists start from far lower baselines (0.75 versus
2.11 monthly languages among high-ability developers), so their
relative expansion is several times larger. Both facts are what the
delegation model predicts: everyone's production broadens with the
agent, but \emph{first uses of unfamiliar languages} concentrate where
unfamiliar candidates are plentiful.

\emph{The ability margin is not cleanly identified by our proxy.}
Low-volume developers show larger measured effects than high-volume
developers, the opposite of the prediction if pre-adoption volume
proxied verification ability. We read this comparison as confounded
rather than as evidence against the mechanism: volume is primarily a
measure of \emph{observed engagement}, and developers with sparse
observed histories mechanically have more languages that can register
as ``new'' --- the same depletion logic behind the attenuation under
stricter activity filters in Section \ref{sec:robust-filters}. The
within-ability specialist comparison is immune to this by
construction; the cross-ability comparison is not, and better ability
proxies (account age, accepted pull requests, organizational
affiliation) are a priority for the next data iteration.

\begin{figure}[H]
\caption{Event studies by heterogeneity cell: language count (top) and
newly-used languages (bottom).}
\label{fig:specialist}
\centering
\includegraphics[width=0.85\textwidth]{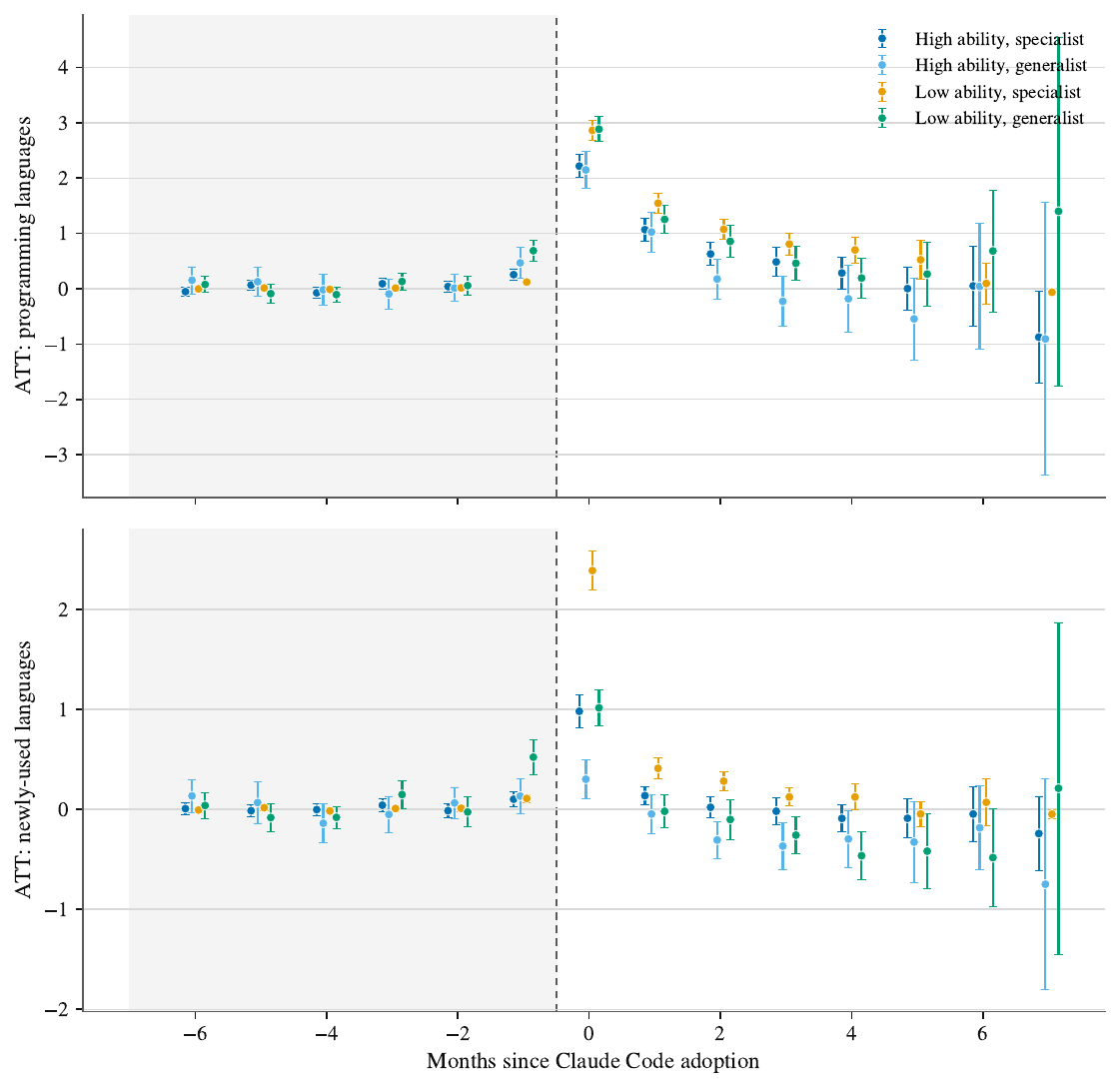}
\end{figure}

The pattern that survives the double sort --- first-use expansion
concentrated among low-breadth developers at every ability level, with
flat pre-trends in all four cells --- is the cross-sectional signature
of delegation into unfamiliar languages, and is difficult to generate
from a uniform activity shock alone.

\section{Discussion: Identification Threats}
\label{sec:discussion}

The estimates in Section \ref{sec:results} are precisely identified
\emph{within} the Callaway--Sant'Anna framework, given the maintained
parallel-trends assumption. The harder question is whether that
assumption is plausible in this setting. We take the threats in turn,
distinguishing those the staggered DiD design addresses from those it
does not.

\subsection{What the staggered DiD design handles}

Three threats are mechanically managed by the estimator. First,
\emph{treatment effect heterogeneity across cohorts}: April 2025
adopters faced a different version of Claude, a different price
schedule, and a different stock of accumulated learning-by-doing than
January 2026 adopters. TWFE specifications can produce sign-flipped
estimates under such heterogeneity \citep{goodmanbacon2021difference};
the cohort-time decomposition in Equation \eqref{eq:attgt} sidesteps
this by estimating cohort-specific effects and aggregating with
non-negative weights. Second, \emph{negative weighting} on
already-treated comparison units: we restrict the comparison group to
the not-yet-treated, ensuring weights are bounded below by zero. Third,
\emph{calendar-time trends}: shocks to GitHub aggregate activity that
hit both groups equally (e.g.\ new platform features, weekend
seasonality, public holidays) are absorbed by the period-specific
$ATT(g, t)$ baseline.

\subsection{What the design does not handle: selection on time-varying
unobservables}

The central threat is that Claude adoption is voluntary and
self-selected. A developer does not adopt Claude at random; she adopts
because something about her work situation has changed. If that
something also affects the outcomes --- the languages she uses, the
repositories she touches --- then the parallel-trends assumption fails
and the estimated ATT confounds the treatment effect with the
selection mechanism.

The most plausible version of this concern is a \emph{reverse-causal}
story at the project level. A developer decides to start a new
project --- perhaps because of a new job, a personal-interest pivot, or
a client requirement --- in a programming language she does not know
well. Because the project is in an unfamiliar language, she installs
Claude Code to assist. Her first Claude commit, our treatment date,
is mechanically simultaneous with the language-portfolio expansion we
measure. In this scenario, Claude does not \emph{cause} the
diversification; both are downstream consequences of the same
upstream decision to start a new project.

The pattern of the $t=0$ coefficient is consistent with this concern.
Across all outcomes the largest effect is the contemporaneous one,
which is exactly what reverse causation predicts: the underlying shock
hits both treatment and outcome at the same moment. Two features of
the data are harder to explain by a one-time project shock alone: the
language count and entropy remain significantly elevated for months
after adoption (including in unassisted commits), and the stock of
distinct languages keeps rising at later horizons (Figure
\ref{fig:event-cumlangs}, Proposition \ref{prop:dynamic}) --- a single
new project delivers its new language once, not a sustained flow. But
the bulk of the $t=0$ effect remains open to the alternative
interpretation.

\subsection{Whose output is it? Joint production}

A distinct objection concerns not selection but attribution: if the
agent writes the code, the measured portfolio might describe the
\emph{tool} rather than the \emph{developer}. Two responses. First, as
a matter of economics, delegated output is joint production, not
autonomous machine output: in the model its value is $\lambda a_i
z(A)$ --- zero for a developer who cannot specify, decompose, and
verify --- and the developer selects which opportunities to take,
carries the verification cost, and ships the result. The frontier we
claim to measure is the developer's \emph{production} frontier
(Remark \ref{rem:production-frontier}), and for that object,
agent-executed work is precisely the margin of interest. Second, as a
matter of data, the decomposition in Section \ref{sec:robust-decomp}
shows the expansion is not even confined to agent-co-authored commits:
roughly two-thirds of the language-count effect appears in the
developer's unassisted commits, which broaden significantly for months
after adoption. The sharpest version of the attribution critique ---
that deleting the agent's commits would leave the developer unchanged
--- is directly rejected by the data.

\subsection{The reversed Ashenfelter dip}

The classical Ashenfelter dip refers to outcomes that fall just before
treatment --- workers' earnings declining in the months before entering
job training, for example. We observe the opposite: monthly commit
activity, language usage, and repository contribution all rise
modestly in the single pre-period month $e = -1$. The
\citet{callaway2021difference} estimator with one-month anticipation
absorbs this rise into the treatment window, but the underlying
pattern suggests that the treatment timing is endogenous: the
developer's overall activity ramps up in the month preceding her first
Claude commit, consistent with the start-of-a-new-project narrative.

\subsection{Robustness specifications and what they rule out}

The checks in Section \ref{sec:robustness} rule out the mechanical
classes of explanation: the effect survives excluding the first-Claude
language, excluding all Claude-co-authored commits, conditioning on
baseline activity, and stricter pre-period activity filters. What
those checks cannot rule out is the project-level selection story
above, in which the same upstream shock drives both adoption and
diversification. That concern is about \emph{why} adoption happens
when it does, not about how the outcomes are constructed, and no
outcome-side scrub addresses it.

\subsection{What would be needed for a causal claim at top venues}

Of the three standard paths toward stronger identification, two are
implemented in this paper and one remains open. \emph{Placebo
treatments}: assigning fake adoption dates twelve months early yields
precise zeros on every outcome (Section \ref{sec:robust-placebo}), so
the design is not picking up spurious timing patterns.
\emph{Conditional parallel trends} with richer covariates: conditioning
the estimator on pre-period commit history, repositories, language
count, and entropy leaves the estimates unchanged (Section
\ref{sec:robust-volume}), so selection on these observables does not
drive the results. What remains open --- and what neither check can
deliver --- is an \emph{exogenous source of variation} in Claude
adoption: regional rollouts of the free tier, pricing-change events, or
eligibility cutoffs in the institutional subscription program would
permit instrumental-variables or regression-discontinuity designs, and
would settle the project-level selection story that outcome-side and
timing-side checks cannot.

\subsection{Honest framing}

In the absence of an exogenous shock, the most defensible
interpretation of the results is associational rather than causal.
The estimates document a \emph{sharp, persistent shift in developer
behavior coincident with Claude Code adoption, in directions
quantitatively consistent with the agentic-delegation mechanism
formalized in Section \ref{sec:theory}}. The pattern is robust to
sample-restriction perturbations and is internally consistent with
the model's prediction of growing dynamic effects on cumulative
language diversity. Whether the underlying mechanism is the
delegated-execution channel of Proposition \ref{prop:frontier} or a
self-selection process that drives both adoption and outcomes will
require the identification work outlined above to resolve.

\section{Conclusion}
\label{sec:conclusion}

Do agentic coding assistants expand the set of programming languages a
developer uses? This paper takes one of the first observational looks at
that question, distinguishing a first generation of AI assistance that
augments work in languages the developer already knows from a second
generation that adds delegated execution. In the model, delegation
lowers the entry threshold for unfamiliar-language opportunities,
activating a band of languages that would not have been used under solo
or conversational production. In a monthly panel of 5{,}346 GitHub
developers with language measured at the commit level from 57 million
changed files, the language portfolio expands sharply around first
detectable Claude Code use: the number of programming languages worked
in rises by 2.5 in the adoption month against a pre-adoption mean of
0.9 and remains elevated for months; developers begin 1.2 languages
absent from their entire observed history; language entropy rises by
0.38; and the cumulative stock of languages keeps growing after
adoption, the stock-flow pattern the dynamic model predicts.
Repository and commit activity also jump, but these are diagnostics of
engagement, not the frontier outcomes.

The expansion is not a bookkeeping artifact. It survives removing the
language introduced by the treatment-defining commit; it survives
rebuilding every outcome from commits without any Claude co-authorship,
where roughly two-thirds of the effect remains; it is unchanged by
conditioning on baseline activity; and it holds among consistently
active developers and after screening out users of competing agentic
tools. The model's sharpest cross-sectional prediction is borne out as
well: first uses of unfamiliar languages concentrate among developers
with narrow pre-adoption portfolios at every ability level --- the
headroom signature of delegation into unfamiliar languages, and a
pattern a uniform activity shock would not produce.

The frontier that moves is a \emph{production} frontier. We do not
claim that adoption teaches a Python developer to write Rust
unassisted; we document that she now ships work in languages she never
touched before, by directing and verifying an agent --- and that her own
unassisted commits broaden too. If these effects are causal, agentic AI
is not merely a productivity multiplier inside a developer's existing
technological boundary but a structural shift in the boundary itself:
an expansion of the horizontal frontier of the individual worker.

We are explicit about what the design cannot settle. Adoption is
voluntary, and its timing is plausibly correlated with decisions to
start projects in unfamiliar languages; the staggered estimator handles
cohort heterogeneity, and the outcome-side checks rule out the
mechanical explanations, but no outcome-side check resolves
project-driven selection into adoption timing. The most defensible
reading is therefore associational: a sharp, robust, dynamic event-time
shift, quantitatively consistent with the agentic-delegation mechanism
of Section \ref{sec:theory}. A definitive causal claim requires
exogenous variation in adoption --- regional rollouts of free-tier
eligibility, pricing-change events, or institutional-subscription
cutoffs --- which we view as the priority for future work.

Several extensions follow directly. Better ability measures --- account
age, accepted pull requests, organizational affiliation --- would
sharpen the heterogeneity analysis, whose ability margin our
volume-based proxy cannot cleanly identify. The order in which
unfamiliar languages are tried --- whether developers expand to
\emph{adjacent} languages first --- maps naturally into sequential
search over a language-similarity graph. And separating delegation from
skill accumulation --- whether the agent remains the executor or the
developer eventually internalizes the language --- requires longer
horizons and content-level measures of who writes what. A separate
strand extends the framework from languages to industry sectors, asking
whether AI-driven broadening enables cross-sector mobility; companion
work \citep{conti2024peer} examines the peer-influence dimension on
collaboration networks.

If AI tools let individual workers operate in domains previously
outside their core human capital, the labor-market consequences of AI
may be less about substitution of capital for labor than about
reallocation of labor across tasks \emph{within} individuals: the
returns to language-specific capital fall, the returns to general
specification-and-verification ability rise, and the matching of
workers to projects changes. The current evidence does not settle the
causal claim, but it documents a large, robust event-time shift that is
difficult to ignore --- and provides a model and an empirical agenda
for identifying this margin cleanly.

\clearpage
\small
\bibliography{references}

\clearpage
\normalsize
\begin{appendices}
\section{Theory Appendix}
\label{app:proofs}

This appendix collects the derivations and proofs behind the propositions
in Section \ref{sec:theory}, states the specialist-heterogeneity and
learning extensions referenced there, and records the repository-expansion
result.

\subsection{CARA--Normal certainty equivalent}

Let a payoff $Y$ be conditionally Normal with mean $m$ and variance
$\sigma^2$, and let utility be $u(y)=-\exp(-\rho y)$ with $\rho>0$. The moment
generating function of the Normal distribution gives
\begin{equation}
\E[-\exp(-\rho Y)]
=-\exp\!\left(-\rho m+\frac{\rho^2\sigma^2}{2}\right).
\end{equation}
The certainty equivalent $\operatorname{CE}$ solves
$-\exp(-\rho\operatorname{CE})=\E[-\exp(-\rho Y)]$, hence
\begin{equation}
\operatorname{CE}
=m-\frac{\rho\sigma^2}{2}.
\label{eq:ce-formula-app}
\end{equation}
This is the mean-minus-risk-penalty expression used in the three production
modes.

\subsection{Threshold algebra}

Starting from the solo surplus
\[
V^S=\omega+s\mu-\frac{\rho s^2}{2\pi}-b,
\]
solo production is viable iff
\begin{equation}
\omega\ge T^S=b-s\mu+\frac{\rho s^2}{2\pi}.
\end{equation}
Generation-1 surplus is $V^C=V^S+\gamma s-r_C$, so its threshold is
\begin{equation}
T^C=b-s\mu-\gamma s+r_C+\frac{\rho s^2}{2\pi}
=T^S-(\gamma s-r_C).
\end{equation}
Because Generation 1 offers either solo or co-pilot production,
\begin{equation}
T^1=\min\{T^S,T^C\}=T^S-\max\{0,\gamma s-r_C\}.
\end{equation}
For an unfamiliar language, Assumption \ref{ass:foothold} implies
$\gamma s-r_C\le0$, so $T^1=T^S$.

Delegation is viable iff
\begin{align}
\omega\ge T^D
={}&
b-(1-\lambda)s\mu-\lambda a z(A)
+\kappa(a,s)+r_D
\nonumber\\
&\quad
+\frac{\rho}{2}
\left[
\frac{(1-\lambda)^2s^2}{\pi}
+\sigma_D^2(a,s,A)
\right].
\end{align}
For an unfamiliar language, the delegation advantage is $B=T^S-T^D$:
\begin{align}
B
={}&
\left(b-s\mu+\frac{\rho s^2}{2\pi}\right)
\nonumber\\
&-
\left\{
b-(1-\lambda)s\mu-\lambda a z(A)+\kappa(a,s)+r_D
+\frac{\rho}{2}
\left[
\frac{(1-\lambda)^2s^2}{\pi}
+\sigma_D^2(a,s,A)
\right]
\right\}
\nonumber\\
={}&
\lambda[a z(A)-s\mu]-\kappa(a,s)-r_D
+\frac{\rho}{2}
\left[
\frac{(2\lambda-\lambda^2)s^2}{\pi}
-\sigma_D^2(a,s,A)
\right],
\end{align}
which is Equation \eqref{eq:agentic-advantage} in the main text.

\subsection{Proof of Proposition \ref{prop:frontier}}

Because $\mathcal{M}_1=\{S,C\}\subset\{S,C,D\}=\mathcal{M}_2$,
\begin{equation}
V^2_{ik,t}
=\max\{V^S_{ik,t},V^C_{ik,t},V^D_{ik,t}\}
\ge
\max\{V^S_{ik,t},V^C_{ik,t}\}
=V^1_{ik,t}.
\end{equation}
If $V^1_{ik,t}\ge0$, then $V^2_{ik,t}\ge0$, so any language active under
Generation 1 remains active under Generation 2. If $V^1_{ik,t}<0$, then
$Z^1_{ik,t}=0$ and $Z^2_{ik,t}\in\{0,1\}$, so $Z^2_{ik,t}\ge Z^1_{ik,t}$.
Summing over languages gives $N^2_{it}\ge N^1_{it}$. Strict inequality holds
when at least one language is inactive under both old modes but active under
delegation. \hfill $\square$

\subsection{Proof of Proposition \ref{prop:band}}

For an unfamiliar language, Assumption \ref{ass:foothold} gives $T^1=T^S$.
If $B=T^S-T^D>0$, then $T^D<T^S$ and the Generation-2 threshold is $T^2=T^D$.
Thus
\[
Z^1=\1{\omega\ge T^S},
\qquad
Z^2=\1{\omega\ge T^D}.
\]
With $T^D<T^S$, these indicators differ exactly when
$T^D\le\omega<T^S$. Therefore
\[
Z^2-Z^1=\1{T^D\le\omega<T^S}.
\]
If $F$ is the conditional CDF of $\omega$, the probability of the band is
$F(T^S)-F(T^D)$. \hfill $\square$

\subsection{Specialist and ability heterogeneity}

The heterogeneity extension referenced in Section \ref{sec:theory} rests
on a symmetry assumption across unfamiliar-language candidates.

\begin{assumption}[Comparable unfamiliar-language candidates]
\label{ass:exchangeable}
Conditional on developer characteristics such as general ability, all
unfamiliar languages have the same per-language activation increment
$p_i(a_i,A)\ge0$, and $p_i$ is increasing in $a_i$ whenever Assumption
\ref{ass:verification} makes delegation thresholds fall with ability.
\end{assumption}

\begin{proposition}[Specialist and ability heterogeneity]
\label{prop:specialist}
Under Assumption \ref{ass:exchangeable}, expected expansion into initially
unfamiliar languages is
\begin{equation}
\E[E_i\mid a_i,U_i]=U_i p_i(a_i,A),
\qquad
E_i\equiv\sum_{k\in\mathcal{U}_i}(Z^2_{ik}-Z^1_{ik}).
\label{eq:specialist-expansion}
\end{equation}
It is increasing in the stock of unfamiliar-language candidates $U_i$ and
in general ability $a_i$. The largest extensive-margin gains accrue to
high-ability specialists.
\end{proposition}

\paragraph{Proof.}
Let $E_i=\sum_{k\in\mathcal{U}_i}(Z^2_{ik}-Z^1_{ik})$. Under Assumption
\ref{ass:exchangeable}, each unfamiliar language has the same activation
increment $p_i(a_i,A)$. Linearity of expectation gives
\begin{equation}
\E[E_i\mid a_i,U_i]
=
\sum_{k\in\mathcal{U}_i}\E[Z^2_{ik}-Z^1_{ik}\mid a_i]
=
U_i p_i(a_i,A).
\end{equation}
Because $U_i=K-|\mathcal{K}_i|$, a smaller familiar-language set mechanically
creates more unfamiliar-language candidates. Assumption \ref{ass:verification}
implies stronger general ability lowers verification costs and weakly lowers
residual error, so $p_i(a_i,A)$ is increasing in $a_i$ when the opportunity
density is positive at the relevant threshold. The product
$U_i p_i(a_i,A)$ is therefore largest for developers with many candidates and
high ability. \hfill $\square$

\subsection{Proof of Proposition \ref{prop:dynamic}}

For an initially unfamiliar language, under generation $g$ the probability of
no first use through horizons $0,\ldots,s$ is
\begin{equation}
\Pr(\mathcal{T}^g_{ik}>s)=(1-p^g_{ik})^{s+1}.
\end{equation}
Therefore the probability of at least one use by horizon $s$ is
\begin{equation}
\Pr(\mathcal{T}^g_{ik}\le s)=1-(1-p^g_{ik})^{s+1}.
\end{equation}
Summing over initially unfamiliar languages and subtracting generations gives
Equation \eqref{eq:cumulative-gap}. If $p^2_{ik}\ge p^1_{ik}$, then
$1-p^2_{ik}\le1-p^1_{ik}$, so each summand is nonnegative.

The first difference of the one-language cumulative gap is
\begin{align}
&\left[(1-p^1)^{s+2}-(1-p^2)^{s+2}\right]
-\left[(1-p^1)^{s+1}-(1-p^2)^{s+1}\right]
\nonumber\\
&\qquad
=
p^2(1-p^2)^{s+1}-p^1(1-p^1)^{s+1}.
\end{align}
Summing over languages, the effect grows between horizons $s$ and $s+1$
if and only if the no-catch-up condition
\begin{equation}
\sum_{k\in\mathcal{U}_i}
\left[
p^2_{ik}(1-p^2_{ik})^{s+1}
-p^1_{ik}(1-p^1_{ik})^{s+1}
\right]\ge0
\label{eq:no-catchup}
\end{equation}
holds. If $p^1=0<p^2$, the first difference is
$p^2(1-p^2)^{s+1}>0$ and the second difference is
$-(p^2)^2(1-p^2)^{s+1}<0$, giving strict increase and concavity.
\hfill $\square$

\subsection{Learning after agentic interaction}

The delegation channel explains why an unfamiliar language can be used
immediately; interaction with the agent and the language can also generate
learning that makes the language persist in the portfolio. Suppose the
developer observes signals
\begin{equation}
x_{\ell}=\theta_{ik}+\varepsilon_{\ell},
\qquad
\varepsilon_{\ell}\sim\N(0,\sigma_{\ell}^2),
\qquad
\ell=1,\ldots,L,
\label{eq:signals}
\end{equation}
conditionally independent given $\theta_{ik}$. Let
$q=\sum_{\ell=1}^L1/\sigma_{\ell}^2$ be the total signal precision and
$\bar{x}_q=(\sum_{\ell}x_\ell/\sigma_\ell^2)/q$ the precision-weighted
signal mean. The posterior kernel is proportional to
\[
\exp\left\{
-\frac{1}{2}\pi(\theta-\mu)^2
-\frac{1}{2}\sum_{\ell=1}^L\frac{(x_\ell-\theta)^2}{\sigma_\ell^2}
\right\},
\]
and collecting terms in $\theta$ gives the Bayesian update
\begin{equation}
\pi'_{ik}=\pi_{ik}+q,
\qquad
\mu'_{ik}=
\frac{\pi_{ik}\mu_{ik}+q\bar{x}_q}{\pi_{ik}+q}.
\label{eq:bayes-update}
\end{equation}
Learning raises precision and updates the productivity belief, so an
initially delegated language can remain active in later months. We use
this channel to interpret persistence and cumulative expansion, but not
as the primary proof of immediate diversification: information can
reveal bad matches as well as good ones.

\subsection{Repository expansion}

\begin{proposition}[Repository expansion]
\label{prop:repos}
Suppose each repository requires at least one programming language and carries
an entry cost that is weakly decreasing when the developer can activate that
language. If agentic delegation weakly expands the active-language set, then
the expected number of repositories the developer can contribute to weakly
increases. It increases strictly when some repositories require languages in
the delegation activation band.
\end{proposition}

\begin{proof}
Let repository $r$ require language $\ell(r)$ and deliver opportunity value
$\Omega_{ir,t}$. Write its entry condition as
$\Omega_{ir,t}\ge c_r(\ell(r),g)$ under generation $g$, where
$c_r(\ell,2)\le c_r(\ell,1)$ whenever delegation weakly lowers the activation
threshold for language $\ell$. The repository indicator under Generation 2
therefore weakly exceeds the Generation-1 indicator for every repository
opportunity. Summing over repositories and taking expectations gives weak
expansion; strict expansion follows if the opportunity distribution places
positive mass on repositories whose required language lies in the activation
band. \hfill $\square$
\end{proof}

\section{Classifying the Tools into the Model's Production Modes}
\label{app:tools}

Table \ref{tab:classification} assigns the AI coding tools of Section
\ref{sec:background} to the model's production modes. The
classification criterion is the \emph{locus of execution}, taken from
each tool's own primary-source capability description, not model
quality or release date: a tool is mode $C$ if the developer must read,
adapt, place, and run every piece of output; it is mode $D$ if the tool
itself edits files, runs commands, and iterates on errors. The 2023
proto-agents are listed separately: they had mode-$D$ mechanics but not
the reliability to constitute a real production mode.

\begin{table}[H]
\caption{Classification of AI coding assistance into the model's
production modes. Mode $C$ tools augment execution the developer still
performs herself ($V^C = V^S + \gamma s - r_C$: value proportional to
existing skill $s$). Mode $D$ tools execute on the developer's behalf
($V^D$: value scales with general ability $a$, not language-specific
skill), which is what moves the entry threshold from $T^1$ to $T^D$.}
\label{tab:classification}
\begin{small}
\begin{singlespace}
\begin{description}[leftmargin=0.45cm,labelsep=0.4em,style=nextline]
\item[\emph{Generation 1 --- autocomplete copilots.}]
Deep TabNine (Jul 2019), GitHub Copilot (Jun 2021), Codeium (2022),
Ghostwriter (Oct 2022), CodeWhisperer (Apr 2023), and Cody (2023)
suggest the next lines or a whole function inside the IDE; the developer
accepts, adapts, and runs the output. These are mode-$C$ tools: they
lower $T^C$ only where skill $s$ is high.

\item[\emph{Generation 1 --- conversational assistants.}]
ChatGPT (Nov 2022), Bing Chat (Feb 2023), Claude (Mar 2023), Bard
(Mar 2023), and Copilot Chat (2023) write snippets, explain errors, and
answer questions in a chat window; the developer transfers everything
into her own environment. These are also mode-$C$ tools, with the same
foothold logic as Assumption \ref{ass:foothold}:
$\gamma s-r_C\le0$ for unfamiliar $k$.

\item[\emph{2023 proto-agents.}]
AutoGPT (Mar 2023), BabyAGI (Apr 2023), smol-developer (May 2023),
GPT-Engineer (Jun 2023), and Aider (mid-2023) had mode-$D$ mechanics
(file writes and iteration) on models that could not sustain them; on
SWE-bench, the best model resolved 1.96\% of tasks. Here $D$ exists, but
$z(A)$ is too low: $T^D>T^1$ for almost all $k$.

\item[\emph{Generation 2 --- agentic assistants.}]
Devin (Mar 2024), SWE-agent (Apr 2024), OpenHands (2024), Copilot
Workspace (Apr 2024), Replit Agent (Sep 2024), Cursor agent (Nov 2024),
Windsurf/Cascade (Nov 2024), Cline (2024), Jules (Dec 2024), Copilot
agent mode (Feb 2025), \textbf{Claude Code (Feb 2025)}, Codex CLI/agent
(Apr--May 2025), Gemini CLI (Jun 2025), and Kiro (Jul 2025) inspect the
repository, edit files, run commands and tests, read errors, and iterate
from a natural-language specification. These are high-$z(A)$ mode-$D$
tools: they create $T^D<T^1$ for unfamiliar $k$, the activation band.

\item[\emph{Adjacent app-builder agents.}]
v0 (Oct 2023), Bolt.new (Oct 2024), and Lovable (Nov 2024) turn a
natural-language specification into a deployed web application. They
have mode-$D$ economics but target non-developers, outside the GitHub
developer population studied here.
\end{description}
\end{singlespace}
\end{small}
\end{table}

Two features of the table do analytical work. First, the mode-$C$ rows
were all available to \emph{every} developer well before the treatment
window, so the pre-period counterfactual is genuinely
$\mathcal{M}_1$-with-mature-tools, not ``no AI.'' Second, the mode-$D$
rows other than Claude Code are exactly the set the competing-agent
screen of Section \ref{sec:data} removes: a developer already using
Devin or Cursor's agent before her first Claude commit already had menu
$\mathcal{M}_2$, so her Claude date would mismeasure the menu
expansion.

\section{Event-Study Support and Extended Pre-Trends}
\label{app:support}

The panel covers 28 months, but the displayable event-study window is
asymmetric by design. On the post side, with not-yet-treated controls,
one-month anticipation, and the last cohort adopting in January 2026
(period 25), group-time effects are estimable only through period 23,
so the earliest cohort reaches at most event time $e=+7$; later horizons
have no remaining comparison units. On the pre side, the earliest cohort
adopts in period 16, so all ten cohorts support pre-period diagnostics
back to $e=-12$ and beyond. Table \ref{tab:risksets} reports the number
of cohorts and developers contributing at each event time; the thinning
support at $e\ge 5$ explains the widening confidence intervals at long
horizons in the main figures.

\begin{table}[H]
\caption{Event-time risk sets: cohorts and developers contributing.}
\label{tab:risksets}
\centering
\begin{threeparttable}
\begin{tabular}{rcc}
\toprule
Event time $e$ & Cohorts contributing & Developers \\
\midrule
$-12$ & 10 & 5,346 \\
$-11$ & 10 & 5,346 \\
$-10$ & 10 & 5,346 \\
$-9$ & 10 & 5,346 \\
$-8$ & 10 & 5,346 \\
$-7$ & 10 & 5,346 \\
$-6$ & 10 & 5,346 \\
$-5$ & 10 & 5,346 \\
$-4$ & 10 & 5,346 \\
$-3$ & 10 & 5,346 \\
$-2$ & 10 & 5,346 \\
$-1$ & 9 & 4,456 \\
$0$ & 8 & 3,823 \\
$1$ & 7 & 3,278 \\
$2$ & 6 & 2,813 \\
$3$ & 5 & 2,204 \\
$4$ & 4 & 1,406 \\
$5$ & 3 & 680 \\
$6$ & 2 & 149 \\
$7$ & 1 & 37 \\
\bottomrule
\end{tabular}
\begin{tablenotes}\footnotesize
\item \emph{Notes.} A cohort $g$ contributes to event time $e$ when
$2 \le g+e \le 23$: the lower bound reflects the varying base period,
the upper bound the last period with not-yet-treated comparison units
under one-month anticipation (last cohort adopts in period 25).
Developer counts are cohort sizes in the estimation sample summed over
contributing cohorts.
\end{tablenotes}
\end{threeparttable}
\end{table}

Figure \ref{fig:extended-pretrend} extends the pre-period window of the
event study to a full year for the two headline diversity outcomes. All
eleven pre-anticipation coefficients ($e \in \{-12,\ldots,-2\}$) are
tightly estimated zeros for both outcomes, with every cohort
contributing at every point; the only non-zero pre-adoption coefficient
remains the anticipation month $e=-1$ discussed in Section
\ref{sec:results-pretrends}.

\begin{figure}[H]
\caption{Extended pre-period window ($e \in [-12, 7]$): programming
languages and language entropy.}
\label{fig:extended-pretrend}
\centering
\includegraphics[width=0.85\textwidth]{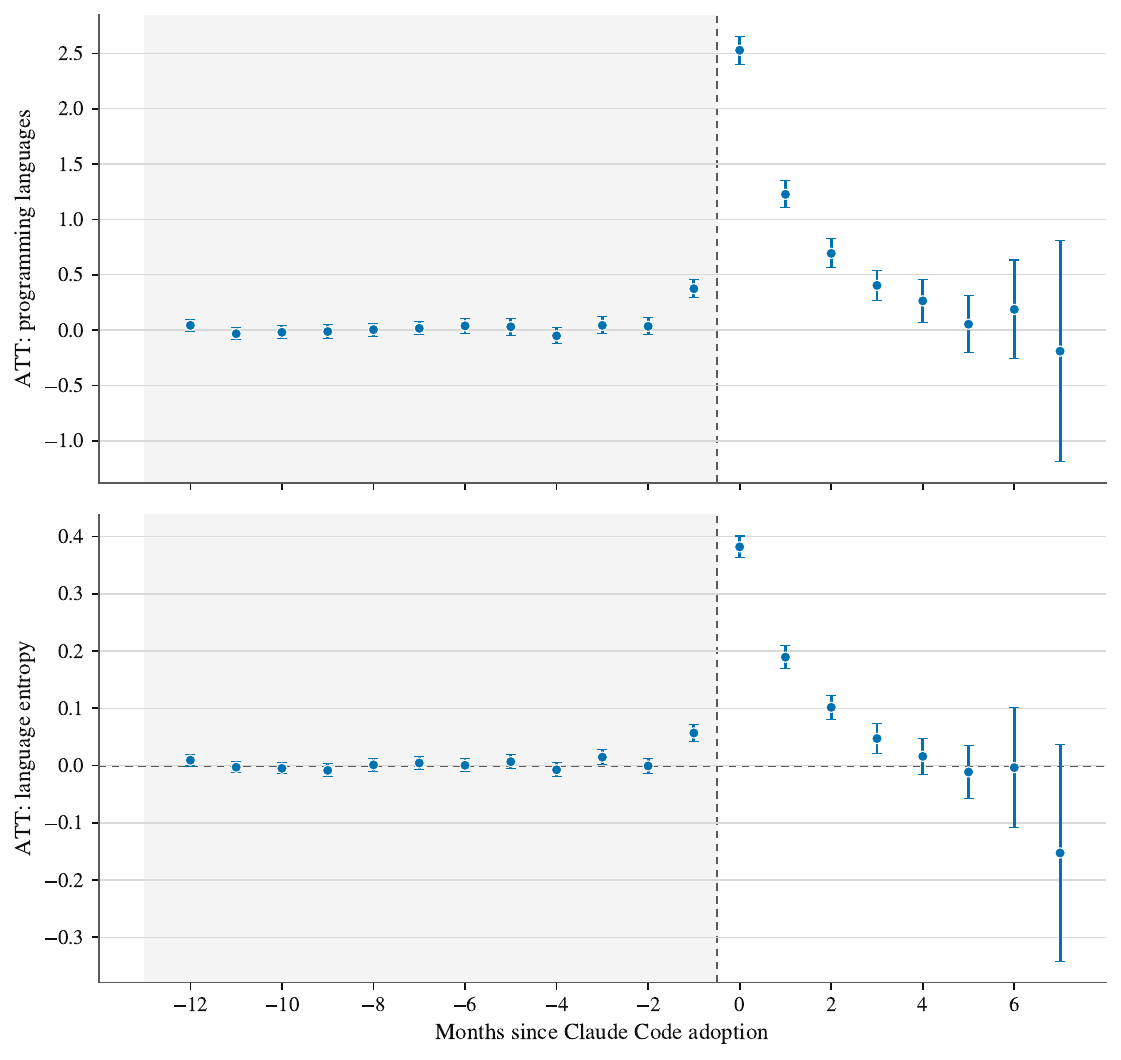}
\end{figure}

\section{Additional Data Construction Details}
\label{app:data}

\subsection*{Bot-account filter}

We exclude developers whose login matches any of the following case-insensitive
substrings, applied after lower-casing the login string: \texttt{[bot]},
\texttt{-bot}, \texttt{\_bot}, \texttt{bot-}, \texttt{dependabot},
\texttt{renovate}, \texttt{github-actions}, \texttt{codecov},
\texttt{greenkeeper}, \texttt{snyk}. The pattern correctly identifies standard
CI/CD bots and dependency managers; we do not attempt to detect bespoke
organizational bots that do not follow these conventions.

\subsection*{Stratification scheme}

The treated sample of 5{,}000 developers is allocated proportionally across the
six adoption-month strata from April through September 2025. Within each month,
the allocation is divided equally across three intensity tiers: \emph{low}
(5--10 Claude commits), \emph{medium} (11--50), and \emph{high} ($\geq51$).
The control sample of 5{,}000 not-yet-treated developers follows the same
stratification scheme, proportional across two adoption-quarter strata
(2025\,Q4 and 2026\,Q1) and equal across the three intensity tiers within each
quarter.

\subsection*{Contribution-API fetch (repository discovery input)}

Each developer's monthly activity is retrieved by querying the GitHub GraphQL
API's \texttt{user.contributionsCollection} endpoint with \texttt{from} and
\texttt{to} parameters set to the first and last calendar dates of each month.
Within each window we retrieve up to 100 repositories
(\texttt{maxRepositories: 100}) and the top 15 languages per repository ordered
by byte size. This fetch serves as one of the three repository-discovery
sources for the commit-level construction below; its language fields are used
only as a cross-check.

\subsection*{Commit-level construction}

The outcome data are built in three stages. \emph{Repository discovery}
unions three sources per developer: the contribution-API repositories
above; the repositories appearing in the raw Claude-commit stream; and
the GH~Archive record of public \texttt{PushEvent}s keyed to the
developer's login, queried on BigQuery over January 2024--April 2026.
Because push events fire for any branch and any authoring email, the
archive source adds 45 percent more (developer, repository) pairs than
the first two sources combined, for a total of 133{,}096 pairs.
\emph{Commit enumeration} lists all commits authored by the developer in
each discovered repository within the window via the REST commits
endpoint, yielding 3{,}151{,}624 distinct commits after de-duplicating
identical commits appearing in both forks and upstream repositories
(each commit is fetched once, preferring the developer's own
repository). \emph{File extraction} obtains each commit's changed file
paths and full commit message from bare, blob-filtered clones of each
repository (\texttt{git clone --bare --filter=blob:none}), reading the
commit graph with \texttt{git log --name-status} across all branches; a
per-SHA REST fallback recovers commits no longer reachable from any
branch after force-pushes. The procedure recovers changed-file records
for 99.997 percent of enumerated commits --- 3{,}151{,}536 commits
touching 57.2 million files. Merge commits are identified by parent
count and excluded from language construction. Each changed file is
assigned a language by GitHub Linguist's filename and extension rules
(exact filenames first, then extensions, ambiguous extensions resolved
to Linguist's popular default), retaining files classified as
\emph{programming} languages and applying Linguist's path-based
vendored, generated, and documentation exclusions. Claude co-authorship
is detected from the \texttt{Co-Authored-By: Claude} trailer in the
commit message; the same message scan detects the competing-agent
signatures listed in Section \ref{sec:data}.

\end{appendices}

\end{document}